\documentclass[11pt]{article}

\usepackage[margin=1in]{geometry}
\usepackage{amssymb,amsmath,amsthm,amsfonts,thmtools,mathtools}
\usepackage{color}
\usepackage[export]{adjustbox}
\usepackage{yfonts}
\usepackage{mathrsfs}
\usepackage{authblk}

\usepackage{xspace}

\usepackage{verbatim}

\usepackage{graphicx,lipsum}
\usepackage[font=small]{caption}
\usepackage{subcaption}
\usepackage[utf8]{inputenc}

\usepackage{url}
\usepackage[all]{xy}
\usepackage{rotating}
\usepackage{ifpdf}
\usepackage{tcolorbox}
\usepackage{lipsum}

\usepackage{mdframed}             
\usepackage{xifthen} 
\usepackage{mathtools}  

\DeclareMathAlphabet\mathbfcal{OMS}{cmsy}{b}{n}

\usepackage{comment}
\usepackage[T1]{fontenc}

\usepackage{array}
\usepackage[shortlabels]{enumitem}

\usepackage{algorithm}
\usepackage[noend]{algpseudocode}

\usepackage{wrapfig} 

\input{insbox}

\newcommand{\ignore}[1]{}

\newcommand{\margincomment}[2]%
{\marginpar{\footnotesize\raggedright {\color{red}#1}: #2}}

\newcommand{\etal}{et~al.}

\newcommand{\myparagraph}[1]{{\medskip\noindent\textbf{#1}}}

\newcommand{\mycase}[1]{{\underline{Case~#1}:}}

\newtheorem{lemma}{Lemma}[section]
\newtheorem{theorem}[lemma]{Theorem}
\newtheorem{corollary}[lemma]{Corollary}

\newtheorem{claim}[lemma]{Claim}
\newtheorem{observation}[lemma]{Observation}

\newcommand{\braced}[1]{{ \left\{ {#1} \right\} }}

\newcommand{\parend}[1]{{ \left({#1} \right) }}

\algnewcommand\algorithmiclet{\textbf{let}}
\algnewcommand\Let{\State \algorithmiclet\ }
\algnewcommand\algorithmiccase{\textbf{case}}
\algdef{SE}[CASE]{Case}{EndCase}[1]{\algorithmiccase\ #1}{\algorithmicend\ \algorithmiccase}%
\algtext*{EndCase}%
\algnewcommand{\IfThenElse}[3]{
  \State
  \algorithmicif\ #1\ \algorithmicthen\ #2\ \algorithmicelse\ #3}

\algnewcommand{\LeftComment}[1]{\Statex \hspace{-1em} \(\triangleright\) \emph{#1}}

\newcommand{\stc}[1]{\mathrm{stc}(#1)}
\newcommand{\cng}[2]{\mathrm{cng}_{#1}(#2)}
\newcommand{\maxcng}[1]{\mathrm{cng}(#1)}

\newcommand{\cutof}[1]{{\partial #1}}
\newcommand{\cutofin}[2]{{\partial_{#2} #1}}
\newcommand{\sptreecut}[2]{#1_{#2}}
\newcommand{\treecutin}[3]{#1^{+#2}_{#3}}
\newcommand{\treecutout}[3]{#1^{-#2}_{#3}}

\newcommand{\Tstar}{\widetilde{T}}

\newcommand{\hubs}{\widetilde H}

\newcommand{\firstweight}{\mathop{\mathrm{w}_1}}
\newcommand{\secondweight}{\mathop{\mathrm{w}_2}}
\newcommand{\dbweight}[2]{#1\!:\!#2}

\newcommand{\bottleneck}{\calB}
\newcommand{\dwgadget}{\calW}
\newcommand{\flower}{\calF}
\newcommand{\rootflower}{\calR}
\newcommand{\twoNflower}{\calH}

\newcommand{\xtwoP}{x^{\mathrm{2P}}}
\newcommand{\xtwoN}{x^{\mathrm{2N}}}
\newcommand{\xthreeP}{x^{\mathrm{3P}}}
\newcommand{\xroot}{x^{\rootflower}}

\newcommand{\yroot}{y^{\rootflower}}
\newcommand{\ztwoP}{z^{\mathrm{2P}}}
\newcommand{\ztwoN}{z^{\mathrm{2N}}}
\newcommand{\zthreeP}{z^{\mathrm{3P}}}
\newcommand{\zroot}{z^{\rootflower}}

\newcommand{\cactusC}{\mathfrakC}

\newcommand{\problemSTC}{{\textsf{STC}}}
\newcommand{\problemKSTC}[1]{{#1\textrm{-}\textsf{STC}}}
\newcommand{\problemMPNSAT}{\mbox{(M2P1N)-SAT}\xspace}
\newcommand{\degree}[2]{\text{deg}_{#1}(#2)}

\newcommand{\NP}{\ensuremath{\mathbb{NP}}\xspace}

\newcommand\dual[1]{\hat{#1}}

\newcommand{\ListLengths}{\setlength{\itemsep}{0ex}\setlength{\topsep}{1ex}\setlength{\partopsep}{0ex}}

\newcommand{\barX}{{\bar X}}
\newcommand{\barY}{{\bar Y}}

\newcommand{\calB}{\ensuremath{\mathcal B}\xspace}

\newcommand{\calF}{\ensuremath{\mathcal F}\xspace}

\newcommand{\calH}{\ensuremath{\mathcal H}\xspace}

\newcommand{\calR}{\ensuremath{\mathcal R}\xspace}

\newcommand{\calW}{\ensuremath{\mathcal W}\xspace}

\newcommand{\tildeO}{{\tilde{O}}}

\newcommand{\mathfrakC}{{\mathfrak{C}}}

\graphicspath{{./FIGURES/}}

\title{Two Complexity Results on Spanning-Tree Congestion Problems\thanks{%
	Research supported by NSF grant CCF-2153723 (M. Chrobak), 
	ANR grant ANR-23-CE48-0010 (C. D\"{u}rr),
	and grant 24-10306S of GA ČR (J. Sgall and P. Kolman).}}

\author[$\dagger$]{Sunny Atalig}
\author[$\dagger$]{Marek Chrobak}
\author[$\diamond$]{Christoph D\"{u}rr}
\author[$\S$]{Petr Kolman}
\author[$\ddagger$] {Huong Luu}
\author[$\S$]{Ji\v{r}{\'\i} Sgall}
\author[$\dagger$] {Gregory Zhu}

\affil[$\dagger$]{University of California at Riverside, USA}
\affil[$\ddagger$]{California Polytechnic University at Pomona, USA}
\affil[$\S$]{Charles University, Prague, Czech Republic}
\affil[$\diamond$]{Sorbonne Universit\'{e}, CNRS, LIP6, France}

\begin{document}

\maketitle

\begin{abstract}
In the spanning-tree congestion problem ($\problemSTC$), given a graph $G$, the objective is to compute
a spanning tree of $G$ for which the maximum edge congestion is minimized. While
$\problemSTC$ is known to be $\NP$-hard, even for some restricted graph classes,
several key
questions regarding its computational complexity remain open, and we address some of these
in our paper. 
(i) For graphs of maximum degree $\Delta$, it is known that
$\problemSTC$ is $\NP$-hard when $\Delta\ge 8$. We provide a complete resolution of this variant, 
by showing that $\problemSTC$ remains $\NP$-hard for each degree bound $\Delta\ge 3$.
(ii) In the decision version of $\problemSTC$, given an integer $K$, the goal is to determine whether the 
congestion of $G$ is at most $K$.
We  prove that this variant is polynomial-time solvable for $K$-edge-connected graphs.
\end{abstract}


\section{Introduction}
\label{sec: introduction}
Constructing spanning trees for graphs under specific constraints is a well-studied problem in graph theory and algorithmics. 
In this paper, we focus on \emph{the spanning-tree congestion problem ($\problemSTC$)}, that arises naturally in some network design and routing problems. 
The problem can be viewed as a special case of the graph sparsification problem where a graph $G$ is 
embedded into its spanning tree $T$ by mapping each edge $(x,y)$ of $G$ to the unique $x$-to-$y$ path in $T$. 
The congestion of an edge $e\in T$ is defined as the number of edges of $G$ 
whose corresponding path in $T$ traverses $e$, and the congestion of $T$ is the maximum congestion of its edges. 
In the $\problemSTC$ problem, we are given a graph $G$ and the objective is to compute
a spanning tree with minimum congestion. This minimum congestion value is referred to as the \emph{spanning-tree congestion of} $G$
and denoted by $\stc{G}$.

The concept of spanning-tree congestion was introduced under different names in the late 
1990s~\cite{sandeep_1986_optimal_tree_machines,simonson_1987_variation_min_cut_arrangement,rosenberg_1988_graph_embeddings,khuller_1993_designing_multi_commodity_flow_tree},
and in 2004 formalized by Ostrovskii~\cite{ostrovoskii_2004_minimal_congestion_tree}, who established some key properties. 
The problem has been extensively studied since then, and
numerous results regarding its graph-theoretic properties and computational complexity
have been reported in the literature. Below we review briefly those that are most relevant to our work.

$\problemSTC$ is $\NP$-hard, with the first $\NP$-hardness proof given by L{\"o}wenstein~\cite{lowenstein_2010_in_the_complement_dominating_set} in 2010. 
It remains $\NP$-hard for planar graphs~\cite{otachi_2010_complexity_result_stc}, chain graphs and split graphs~\cite{okamoto_2011_hardness_results_exp_algorithm_stc}. 
On the other hand, $\problemSTC$ is polynomial-time solvable for a wide variety of special graph classes, including
complete $k$-partite graphs, two-dimensional tori~\cite{kozawa_2009_stc_graphs}, outerplanar graphs~\cite{bodlaender_2011_stc_k-outerplanargraphs}, 
two-dimensional Hamming graphs~\cite{kozawa_2011_stc_rook_graphs}, co-chain graphs~\cite{kubo_2015_spanning_tree_small_parameter}, and interval graphs~\cite{lin_2025_stc_interval_graphs}. 

In the decision version of $\problemSTC$, in addition to a graph $G$ we are also given an integer $K$, and the goal is to determine if $\stc{G} \le K$. 
A natural variant of $\problemSTC$, when the congestion parameter $K$ is a fixed constant (rather than given as input) is denoted $\problemKSTC{K}$.  
The $\problemKSTC{K}$ problem was shown to be $\NP$-complete for $K \ge 5$ by 
Luu and Chrobak~\cite{luu_chrobak_2025_better_hardness_algo}, building on earlier results 
for larger constants~\cite{otachi_2010_complexity_result_stc,bodlaender_2012_parameterized_complexity_stc}. 
On the other hand, $\problemKSTC{K}$ is solvable in linear time for $K\le 3$~\cite{otachi_2010_complexity_result_stc}. 
The complexity status of $\problemKSTC{4}$ remains an intriguing open problem. 
For graphs of radius $2$, $\problemKSTC{K}$ is $\NP$-complete for $K\ge 6$~\cite{luu_chrobak_2025_better_hardness_algo},
and its complexity is open for $K=4,5$.  
For any constant $K$, $\problemKSTC{K}$ is linear-time solvable for bounded-degree graphs, bounded-treewidth graphs, apex-minor-free graphs~\cite{bodlaender_2012_parameterized_complexity_stc}, 
and chordal graphs~\cite{otachi_2020_survey_spanning_tree_congestion}. 

Kolman~\cite{kolman_iwoca_2024} observed that the existing $\NP$-hardness proofs used graphs
of unbounded degree, and raised the question about the complexity of $\problemSTC$
for graphs of constant degree. For constant-degree graphs, it has only been known
that $\problemKSTC{K}$ is linear-time solvable if the congestion bound $K$ is also constant~\cite{bodlaender_2012_parameterized_complexity_stc}.
Recently, Lampis~\etal~\cite{lampis_etal_parameterized_spanning_tree_congestion_2025}
reported progress on this problem, by proving that,
for any constant $\Delta\ge 8$, $\problemSTC$ is $\NP$-hard for graphs with maximum degree $\Delta$, 
leaving open the complexity of  $\problemSTC$ for degree bounds between $3$ and $7$.


\myparagraph{Our contributions.} 
Addressing the problem left open by Lampis~\etal~\cite{lampis_etal_parameterized_spanning_tree_congestion_2025}, 
we prove (see Section~\ref{sec: np-hardness for degree-3 graphs}) the following theorem:

\begin{theorem}\label{thm: np-completeness for degree 3}
Problem $\problemSTC$ is $\NP$-hard for graphs of maximum degree at most $3$.
\end{theorem}

Naturally, this theorem is true for all degree bounds $\Delta\ge 3$, and it remains true for $3$-regular graphs.
(Vertices of degree $1$ or $2$ can be removed from the graph, without affecting its maximum congestion value.)
This result fully resolves the status of $\problemSTC$ for bounded-degree graphs. 

\smallskip

We also  study the $\problemKSTC{K}$ problem for $K$-edge-connected graphs (see Section~\ref{sec: linear-time algorithm}),
proving the following theorem:

\begin{theorem}\label{thm: algorithm for k-edge-connected}
There is an $\tildeO(m)$-time algorithm that,
given a $K$-edge-connected graph $G$, determines whether
$\stc{G}=K$.
\end{theorem}

Above, $m$ is the number of edges, and
the $\tildeO(m)$ time bound is actually independent of $K$.
Our solution is based on the so-called cactus representation of $K$-cuts in $K$-edge-connected graphs
that was developed by Dinic~\etal~\cite{dinitz_etal_strukture_systemy_1976} (see also~\cite{Fleiner_Frank_cactus_mincuts_2009}).
We further refine this characterization for graphs with congestion $K$, to obtain additional properties that lead to
a fast algorithm.
Besides its own interest, this result sheds new light on the complexity of $\problemKSTC{4}$, 
showing that its difficulty is related to the presence of cuts of size less than $4$ in the graph.


\myparagraph{Other related work.}
General bounds for the spanning-tree congestion value have been well studied.
For graphs with $n$ vertices and $m$ edges it is known that the congestion is at most $O(\sqrt{mn})$
and that there are graphs where this value is $\Omega(\sqrt{mn})$~\cite{chandran_et_al_spanning_tree_congestion_2018}.
Since the congestion of an $n$-clique is $n-1$, this implies that, somewhat counter-intuitively, the congestion value is 
not monotone: adding edges can actually decrease the congestion, and quite substantially so. 

This non-monotonicity is particularly challenging in the context of approximations. Indeed, very little is known 
about the approximability of $\problemSTC$. While the upper bound of $n/2$ on the approximation ratio 
is trivial (achieved by \emph{any} spanning tree~\cite{otachi_2020_survey_spanning_tree_congestion}), the best known lower bound is only $1.2$,
implied directly by the $\NP$-completess of $\problemKSTC{5}$~\cite{luu_chrobak_2025_better_hardness_algo}. 
In a recent work,  Kolman~\cite{kolman_approximating_congestion_2025} developed
an algorithm with approximation ratio $\tildeO(\Delta)$, where $\Delta$ is the maximum vertex degree.
Yet the general problem remains open; in particular it is not known if it is possible to achieve ratio $O(n^\delta)$, for some $\delta < 1$.

The spanning-tree congestion is related to the tree spanner problem which seeks a spanning tree with minimum stretch factor. 
The two problems are in fact equivalent in the case of planar graphs: the spanning-tree congestion of a planar graph 
is equal to the minimum stretch factor of its dual plus one~\cite{fakete_2001_tree_spanner,otachi_2010_complexity_result_stc} 
(cf.~\cite{cai_1995_tree_spanner,fakete_2001_tree_spanner,dragan_2011_spanner_in_sparse_graph} and the references therein for
further discussion). It is worth mentioning here that
the complexity status of the tree 3-spanner problem has remained open since its introduction in 2014~\cite{cai_1995_tree_spanner}.

The $\problemSTC$ problem can be relaxed by dropping the restriction that the tree to be computed
is a spanning tree of the graph. In this version, the tree
must include all vertices, but its edges do not need to be present in the underlying graph. This variant
arises in the context of multi-commodity tree-based routing~\cite{seymour_1994_call_routing,khuller_1993_designing_multi_commodity_flow_tree},
and appears to be computationally easier that $\problemSTC$, as it admits an $O(\log n)$ approximation.

Readers interested in learning more about the $\problemSTC$ problem are referred to
the survey by Otachi~\cite{otachi_2020_survey_spanning_tree_congestion} that covers the state-of-the-art as of 2020, and
to the recent paper by Lampis~\etal~\cite{lampis_etal_parameterized_spanning_tree_congestion_2025} that has
additional information about some recent work, in particular about the parametrized complexity of $\problemSTC$.


\section{Preliminaries}
\label{sec: preliminaries}

Throughout the paper, by $G = (V,E)$ we denote an undirected graph with 
vertex set $V$ and edge set $E$.  
For a vertex $u$, by $E_u$ we denote the set of edges incident with $u$ and by
$N_u$ or $N(u)$ we denote the set of $u$'s neighbors. We extend this
notation naturally to sets of vertices.
For a set of edges
$E'\subseteq E$, $V(E')$ denotes the set of all vertices incident with some edge of $E'$. 

Recall that a connected graph $G = (V,E)$ is said to be \emph{$K$-edge-connected} if
it remains connected even after removing $K-1$ edges. 
For any subset $X\notin\braced{\emptyset,V}$ of vertices, by $\cutof{X}$ (or $\cutof{\barX}$) we denote the set of
edges between $X$ and $\barX=V\setminus X$, and we call it a \emph{cut}.
We refer to $X$ and $\barX$ as the \emph{shores} of cut $\cutof{X}$.
If $|\cutof{X}| = K$, we say that $\cutof{X}$ is a \emph{$K$-cut},
and if $|X|=1$ or $|\barX| = 1$ then we call cut $\cutof{X}$ \emph{trivial}.

If $T$ is a spanning tree of $G$ and $e$ is an edge of $T$, removing $e$ from $T$ disconnects $T$ into two
connected components. Given a vertex $x$ of $T$, we denote the
component containing $x$ by $\treecutin{T}{x}{e}$ and the component not
containing $x$ by $\treecutout{T}{x}{e}$ (we interpret these as sets of vertices of $G$).
The cut $\cutof{\treecutin{T}{x}{e}}=\cutof{\treecutout{T}{x}{e}} = \cutof{\sptreecut{T}{e}}$ is called \emph{the cut induced by $e$}.
(So $\treecutin{T}{x}{e}$ and $\treecutout{T}{x}{e}$ are its two shores.)
Its cardinality $|\cutof{\sptreecut{T}{e}}|$ is called the \emph{congestion of $e$ in $T$} and is
denoted by $\cng{G,T}{e}$, or $\cng{T}{e}$ if $G$ is understood from context.
The \emph{congestion of tree $T$}, denoted $\maxcng{G,T}$, is the maximum
edge congestion in $T$. 
The minimum value of $\maxcng{G,T}$ over all spanning trees $T$ of $G$ is called
the \emph{spanning-tree congestion of $G$} and is denoted $\stc{G}$.
It is easy to see that this definition, expressed in terms of induced cuts, is
equivalent to the definition given at the beginning of Section~\ref{sec: introduction}.


\section{$\NP$-Hardness for Degree-$3$ Graphs}
\label{sec: np-hardness for degree-3 graphs}

In this section we prove Theorem~\ref{thm: np-completeness for degree 3}.
Our proof is via a polynomial-time reduction from an $\NP$-complete 
version of SAT (defined below), mapping a boolean expression $\phi$
into a graph $G$ and integer $K$, such that $\stc{G}\le K$ if and only if $\phi$ is satisfiable.
$G$ consists of multiple \emph{gadget} subgraphs, some corresponding to variables and some to clauses,
as well as one additional \emph{root gadget}, with appropriate connections in and between these gadgets. 
Some gadgets are constructed from smaller \emph{sub-gadgets}. 
The most basic gadget is called a \emph{double-weight gadget}, and it allows us to use
edges that are assigned two weight values, with appropriate interpretation.
Using double-weighted edges, we construct a more complex \emph{flower gadget}
that will be used as the root gadget and as the gadgets for some clauses.
The restriction to degree $3$ makes these constructions quite intricate.
A considerably simpler (through structurally similar) proof for graphs of degree at most $4$ can be found in Appendix~\ref{sec: np-hardness degree 4}.


\myparagraph{Problem $\problemMPNSAT$.} 
This is an $\NP$-complete restriction of SAT~\cite{luu_chrobak_2025_better_hardness_algo}, 
whose instance is a boolean expression in conjunctive normal form with the following properties:
\begin{description}[nosep]
\item{(i)} each clause contains either three positive literals (3P-clause), or two positive literals (2P-clause), or two negative literals (2N-clause), and
\item{(ii)} each variable appears exactly three times, exactly once in each type of clause,
	and any two clauses share at most one variable.
\end{description}

We use the following conventions: letter $\phi$ is an instance of $\problemMPNSAT$,
boolean variables are denoted with Latin letters $x,y,z$, while for clauses we use Greek letters
$\kappa$, $\alpha$, $\beta$, $\gamma$ and $\pi$. Typically,
$\alpha$, $\beta$ and $\gamma$ denote a 2N-clause, 3P-clause and 2P-clause, respectively,
$\kappa$ is a clause of any type, and $\pi$ is a positive clause.
For a variable $x$, by the \emph{2N-clause of $x$} we mean the unique
2N-clause that contains the negative literal of $x$. Similarly, the \emph{2P-} and \emph{3P-clauses of
$x$} are the unique 2P- and 3P-clauses, respectively, that contain the positive literal of $x$.


\myparagraph{Double weights.}
In the graph we construct from an instance of $\problemMPNSAT$
we will need the notion of double-weighted edges, introduced by Luu and Chrobak~\cite{luu_chrobak_2025_better_hardness_algo}. 
Consider a pair of weight functions $\firstweight, \secondweight : E \to \braced{1,...,M}$, 
where $M$ is a positive integer whose value is polynomial in $n$.
Given a spanning tree $T$ and an edge $e=\parend{u,v}\in T$, we define the \emph{weighted congestion} of $e$ to be
\begin{equation*}
\cng{G,T} {e} \;=\; \textstyle \secondweight(e) + \sum_{e' \in \cutof{\sptreecut{T}{e}} \setminus \braced{e}} \firstweight(e') \ .
\end{equation*}
The definitions of $\maxcng{G,T}$ and $\stc{G}$ extend naturally to double-weighted graphs.
We also define the \emph{weighted degree} of $v$ to be $\degree{G}{v} = \sum_{(u,v) \in E} \firstweight(u,v)$. 
(Note that it depends only on weight function $\firstweight$.)
As long as the meaning is clear from context, we will often drop the word \emph{weighted} and
write simply \emph{congestion} or \emph{degree}, while we mean the weighted versions of these terms.

When $\firstweight(e) = \secondweight(e) = 1$, we say 
that $e$ is \emph{unweighted}. For a double-weighted edge $e$ with $(a,b)= (\firstweight(e),\secondweight(e))$, 
we write its weight as $\dbweight{a}{b}$.

We will only use weight functions such that $\firstweight(e)\le \secondweight(e)$ for each $e\in E$,
so one can think of $\firstweight$ as the \emph{light weight} function and $\secondweight$ as the \emph{heavy weight} function.
An important intuition is that the edges in $T$ contribute their heavy weight to their induced cuts, while other edges
contribute their light weight to the cuts they belong to. This is why using double-weight edges gives us more control over the
structure of trees with optimal congestion, thus greatly simplifying the construction of our target graph.

An edge $e=(u,v)$ with $\firstweight(e) = \secondweight(e)$ can be replaced by $\firstweight(e)$ edge-disjoint paths from $u$ to $v$ without affecting
the congestion or vertex degrees.
As shown by Luu and Chrobak~\cite{luu_chrobak_2025_better_hardness_algo}, this idea can be extended to \emph{unequal} double weights: in a graph $G$, 
an edge with weight $\dbweight{a}{b}$ (under some mild assumptions) can be replaced by an appropriate unweighted sub-graph called a \emph{double-weight gadget},
so that the resulting graph $G'$ has the property that $\stc{G} \le K$ if and only if $\stc{G'} \le K$.
In their construction,
the maximum degree of $G'$  becomes large if $b$ is large. 
Lampis~\etal~\cite{lampis_etal_parameterized_spanning_tree_congestion_2025} provide a low-degree double-weight 
gadget, but it is not sufficient for our purpose (because it contains vertices of degree $4$).

In Appendix~\ref{sec: double weight gadget}, we construct a degree~$3$ double-weight gadget $\dwgadget(a,b), s^\ast, t^\ast$,
where $a,b$ represent a double weight $\dbweight{a}{b}$ and $s^\ast, t^\ast$ are two designated vertices called ``ports''.
There we also prove that $\dwgadget(a,b)$ has the following properties.


 
\begin{lemma}\label{lemma: double weight gadget}
Let $K\ge 3$ be an integer, $G$ a double-weighted graph of maximum degree $\Delta\ge 3$, 
and $e=(u,v) \in E$ an edge in $G$ with double weight $\dbweight{a}{b}$ satisfying $a<b$ and $b-a \le K-2$. 
Let $G'$ be the graph obtained from $G$ by removing $e$ and replacing it by a degree-$3$ double-weight
gadget $\dwgadget(a,b)$
whose ports $s^\ast$ and $t^\ast$ are identified with $u$ and $v$, respectively.
Then 
\begin{description}[nosep]
	\item{\normalfont{(i)}} The (weighted) degrees of the original
		vertices of $G$ (including $u,v$) remain unchanged. Thus the maximum degree in $G'$ is $\Delta$.
	\item{\normalfont{(ii)}} $\stc{G} \le K$ if and only if $\stc{G'} \le K$. 
\end{description}
\end{lemma}


\myparagraph{Flowers.} 
One can think of our construction as having an intermediate implicit step, where variables and clauses are mapped
into vertices of an auxiliary graph $G^\ast$ that also contains a high-degree \emph{root} vertex.
To obtain the final graph $G$, the high-degree vertices are replaced by other appropriate gadgets.
One gadget is called a \emph{flower}, and it will be used to replace 2N-clause vertices and the root vertex of $G^\ast$.
One other (unnamed) gadget, with a slightly different functionality, will be used as the variable gadget.

An \emph{$\ell$-terminal flower gadget with integrality $K$}, denoted $\flower(\ell,K)$, is the following graph:
\begin{itemize}[nosep]
\item It has $3\ell$ vertices $c_1, \ldots, c_{\ell}$, $d_1, \ldots, d_{\ell}$, and $t_1, \ldots, t_{\ell}$, 
	that we call the \emph{core}, \emph{dummy}, and \emph{terminal} vertices, respectively.  Vertex $c_1$ is special and is designated as the \emph{center} of the flower.
\item For each $i\in [\ell]$\footnote{By $[\ell]$ we denote the set $\{1,2,\ldots,\ell\}$.}, it has
 the following edges (indexing is cyclic, so $\ell+1=1$):
(i) $\parend{c_i, c_{i+1}}$, $\parend{t_i, d_i}$ and $\parend{t_i, d_{i+1}}$, all with weight $1$,
(ii) $\parend{c_i, d_i}$ with weight $\dbweight{1}{K-1}$ for $i\neq 1$ and weight $1$ for $i=1$.
\end{itemize}

\begin{figure}[ht]
\centering
\includegraphics[width=2in]{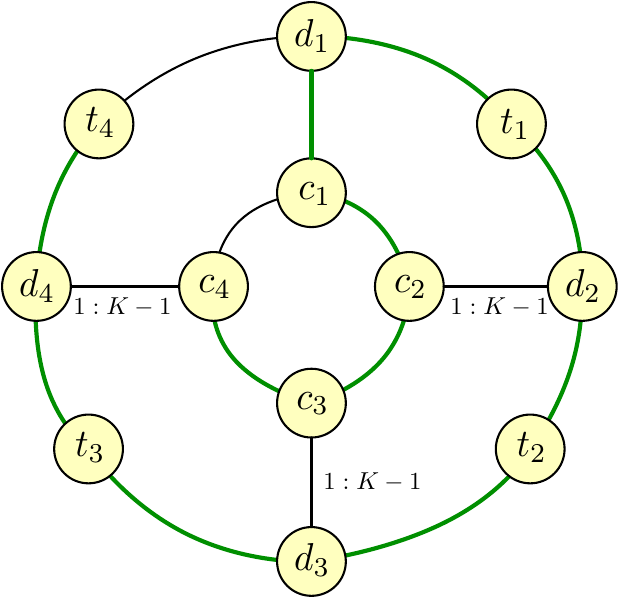}
\caption{A $4$-terminal flower $\flower(4,K)$, and its congestion-$5$ spanning tree marked with thick (green) lines.} \label{fig:4 leaf flower}
\end{figure}

Figure~\ref{fig:4 leaf flower} shows $\flower(4,K)$.
Notice that the core and dummy vertices have degree $3$, while the terminal vertices have degree $2$ (these will be used to attach
the flower to the rest of the graph via edges with unit $\firstweight$-weights).
$\flower(\ell,K)$ has $O(\ell)$ vertices and edges, and it has a spanning tree with congestion $\ell+1$ given by edges
$\parend{c_i, c_{i+1}}$, $\parend{t_i, d_i}$, $\parend{t_i, d_{i+1}}$,
for $i\in[\ell-1]$, and $\parend{t_{\ell}, d_{\ell}}$, $\parend{d_1, c_1}$.

Expanding on the intuition outlined earlier, to simulate a high-degree vertex, a flower should ideally
have the property that a spanning tree with congestion $K$ can
visit it only once in the sense that the flower lies wholly inside one shore of any cut induced by any non-flower edge in this tree.
As this is difficult to achieve using degree-$3$ vertices, we instead require only the core to lie on one 
shore of such cuts. 
To make up for this relaxation, the flower has the property that disjoint paths touching the 
gadget's terminals can be extended to disjoint paths touching the core. This disjoint-paths property will be crucial in
the proof.

Let $G$ be a double-weighted graph, and $U\subsetneq V$ a set of vertices that induces a flower subgraph $\flower = \flower(\ell,K) = (U,F)$
such that the cut $\cutof U$ consists of exactly $\ell$ edges, each connected to a different terminal in $\flower$.
The intuition above is formalized below.


\begin{observation}\label{obs: integrity of core}
For any spanning tree $T$ of $G$ with $\cng{G}{T} \le K$ and an edge $e \in T \setminus F$, all the core vertices of $\flower$ 
belong to the same shore of the cut $\cutof{\sptreecut{T}{e}}$.	
\end{observation}

\begin{proof}
Observe first that $T$ contains no edges between core and dummy vertices of $\flower$ except $(d_1, c_1)$:
if $(d_i, c_i) \in T$ for some $i\ne 1$, then the cut induced by $(d_i, c_i)$ would have congestion $K+1$
(because $(d_i,c_i)$ itself contributes $K-1$, and there are two edge-disjoint paths from $c_i$ to $d_i$ within $\flower$ 
not containing $(d_i, c_i)$), contradicting the assumption that $\cng{G}{T} \le K$.
Thus, $T$ contains all but one edge from the core $c_1-c_2-...-c_\ell-c_1$, which implies the observation.
\end{proof}


Note that in the flower $\flower(\ell,K)$, the $\ell$ length-$2$ paths $t_i-d_i-c_i$ are disjoint.
Thus, given a collection of edge-disjoint paths $\braced{P_i}_{i \in I}$ in $G$ indexed by $I\subseteq [\ell]$, 
where each path $P_i$ starts at some $u_i$ in $V-U$, 
ends at terminal $t_i$ of $\flower$ and does not contain any other vertex of $\flower$, 
we can easily extend it to a collection $\braced{P'_i}_{i \in I}$ of edge-disjoint paths
with each $P'_i$ starting at $u_i$ and ending at $c_i$.


\myparagraph{The reduction.} 
Given an instance $\phi$ of $\problemMPNSAT$, we convert it into
a double-weighted graph $G$ with maximum degree $3$, and a constant $K$ such that:
\begin{description}\setlength{\itemsep}{-0.03in}
	\item{$(\ast)$} the boolean expression $\phi$ is satisfiable if and only if $\stc{G}\le K$.
\end{description}
The weights in $G$ will be bounded by a polynomial function of the size of $\phi$; thus,
per Lemma~\ref{lemma: double weight gadget} and the construction of the double-weight gadget 
$\dwgadget(a,b)$ in Appendix~\ref{sec: double weight gadget}, 
this reduction will be sufficient to establish Theorem~\ref{thm: np-completeness for degree 3}.

Let $K\ge 6$ be a positive integer to be specified later, and let 
$n, m, m_1$ and $m_2$, resp., be the number of variables, clauses, 2N-clauses and 2P-clauses of $\phi$, resp.
We construct a double-weighted graph $G=(V,E; \firstweight, \secondweight)$ as follows (cf.~Fig.~\ref{fig: structure of G}):

\begin{itemize}[nosep]

\item
Create a copy $\rootflower$ of $\flower(2m_1+m_2+n,K)$, that we call \emph{the root} or \emph{root-gadget}:
it has one terminal $t_\rootflower^x$ per each variable $x$, 
one terminal $t_\rootflower^\gamma$ per each 2P-clause $\gamma$, 
and two terminals $t_\rootflower^{\alpha, 1}$ and $t_\rootflower^{\alpha, 2}$ per each 2N-clause $\alpha$. 
(How exactly we assign terminals to clauses/variables is irrelevant.) 

\item
For each positive clause $\pi$, create a vertex $\pi$. For every 2N-clause $\alpha$, 
create a flower $\twoNflower_\alpha = \flower(4,K)$, called the \emph{$\alpha$-gadget} or  \emph{clause-gadget},
and denote its terminals by $t_\alpha^{\rootflower, 1}, t_\alpha^{\rootflower, 2}, t_{\alpha}^x, t_\alpha^y$, where $x$ and $y$ are the variables in $\alpha$.
(How we assign labels to terminals is irrelevant). 

\item
For each variable $x$, create a length-4 cycle $\xtwoN-\xtwoP-\xroot-\xthreeP-\xtwoN$ with unweighted edges, 
called a \emph{variable-gadget} or \emph{$x$-gadget}. 
(The order of vertex labels in this cycle \emph{is important}.)
Then add a \emph{root-variable edge} $(\xroot, t_\rootflower^x)$ with weight $\dbweight{1}{K-5}$,
and \emph{clause-variable} unweighted edges $(t_\alpha^x, \xtwoN)$, $(\beta, \xthreeP)$, $(\gamma, \xtwoP)$, 
where $\alpha, \beta, \gamma$ are the 2N-clause, 3P-clause, and 2P-clause of $x$, respectively.

\item
For each 2P-clause $\gamma$, add one \emph{root-clause} edge $(\gamma, t_\rootflower^\gamma)$ with weight $\dbweight{1}{K-1}$. 
For each 2N-clause $\alpha$, add two \emph{root-clause} edges $(t_\alpha^{\rootflower, 1}, t_\rootflower^{\alpha, 1})$ and $(t_\alpha^{\rootflower, 2}, t_\rootflower^{\alpha, 2})$, 
each of weight $\dbweight{1}{K-1}$. 
(Note that there are no edges between 3P-clases and the root gadget.)

\end{itemize}

\begin{figure}[ht]
\begin{center}
	\includegraphics[width=5in]{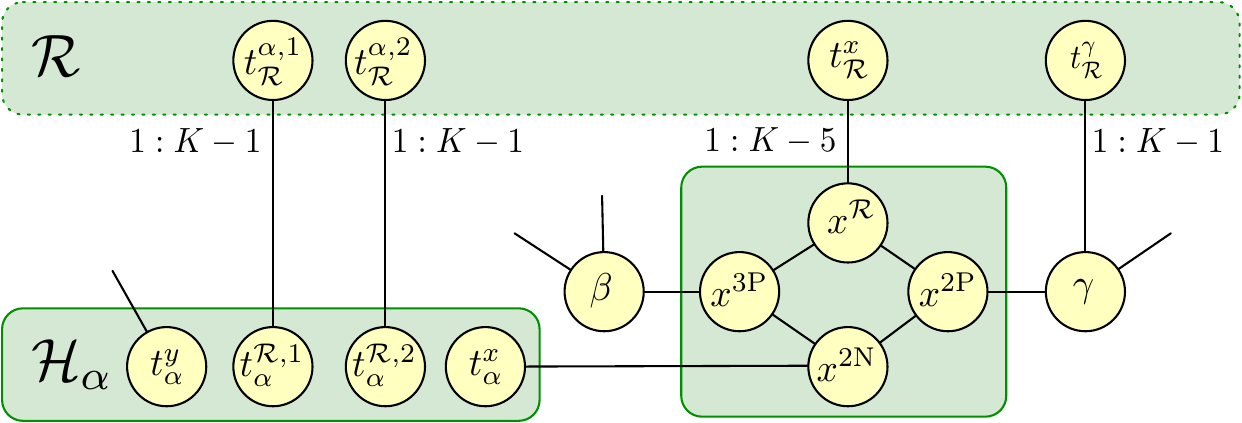}
    \caption{The structure of $G$.}
	\label{fig: structure of G}
\end{center}
\end{figure}

\smallskip

By inspection, all vertices in this construction have degree 3, and $\firstweight(e) =1$ for all $e\in E$. 
The number of edges $|E|$ in $G$ is independent of $K$, so we can set $K=2|E|$. Our double-weighted $\problemSTC$ instance $G$ is now fully specified. 

Any edge belonging to a variable, clause, or root gadget 
is called an \emph{internal edge}, while every other edge is \emph{external}. 
We identify flowers by the corresponding clause (or root); for example by \emph{core of clause $\alpha$} we mean the core of $\twoNflower_\alpha$. 
(Similar terminology applies to flower terminals and centers.) For the sake of uniformity, 
we also refer to positive-clause vertices as \emph{centers}.
In the proof it will be convenient to occasionally work in the auxiliary multi-graph $G^*$ given by contracting all internal edges of $G$,
with each flower and variable gadget contracted to a single vertex, for which we use the same notation as the gadget itself.
We extend this convention to the edges in $G^*$: for example a root-clause edge $(t_\alpha^{\rootflower, 1},t_\rootflower^{\alpha, 1})$ in 
$G$ between a terminal of clause $\alpha$ and a terminal of the root is represented by edge $(\alpha, \rootflower)$ in $G^*$ between $\rootflower$ and $\alpha$.


\myparagraph{Correctness.}
We now need to prove that our construction is correct, namely that it satisfies the condition $(\ast)$. We prove the two
implications in $(\ast)$ separately.

\medskip\noindent
($\Rightarrow$)
Given a satisfying assignment for $\phi$, we convert it into a spanning tree $T$ for $G$ as follows:
(1) for each gadget, add a spanning-tree of unweighted edges for this gadget to~$T$ (one always exists for flowers and 4-cycles),
(2) add all root-variable edges,
(3) for each clause $\kappa$, pick any (exactly one) variable $x$ whose assignment satisfies $\kappa$ 
			and add the clause-variable edge from $\kappa$ to $x$ to $T$.

It is easy to see that $T$ is a spanning tree for $G$. Furthermore, we can obtain a spanning tree $T^*$ for $G^*$ by contracting the internal edges in $G$, and every clause is a leaf in $T^*$. We now show that $\cng{G,T}{e} \le K$ for all $e \in T$.

When $e\in T$ is unweighted, we trivially have $\cng{G,T}{e} \le |E| \le K$. Otherwise $e=\parend{x, t_\rootflower^x}\in T$ is an external root-variable edge, 
in which case the vertices for any given gadget will be together on one shore of the cut induced by $e$, implying $\cng{G,T}{e}=\cng{G^*, T^*}{x,\rootflower}$. 
Thus, it suffices to deal with congestion in $G^*$.

Let $\alpha, \beta, \gamma$ be the 2N-clause, 3P-clause, and 2P-clause containing $x$ respectively. 
One shore of the cut induced by $(x,\rootflower)$ is given either by 
(i) $\braced{x}$,
(ii) $\braced{\kappa, x}$ for $\kappa \in \braced{\alpha, \beta, \gamma}$, or 
(iii) $\braced{x, \beta, \gamma}$. 
In case (i), $\cng{G^*, T^*}{x,\rootflower} = K-5+\degree{G^*}{x}-1  \le K-2$.
In case (ii), $\cng{G^*, T^*}{x,\rootflower}=K-5+3 +2 \le K$, since $\kappa$ is incident to at most 3 edges in the cut and $x$ is incident to 2 edges besides $(\rootflower,x)$.
In case (iii), 
$\cng{G^*,T^*}{x,\rootflower} = K-5 +  2 + 2 + 1 = K$, since both positive clauses are incident to 2 cut edges while $x$ is incident to 1 besides $(\rootflower,x)$.


\medskip\noindent
($\Leftarrow$)
Given a spanning tree $T$ for $G$ with $\maxcng{G,T} \le K$, our goal is to construct a satisfying assignment for $\phi$. Ideally, $T$ would have a form similar to the spanning tree described in the $\Rightarrow$~direction, but this may not be the case. 
 Nevertheless, $T$ has enough structure, captured in Lemmas~\ref{lemma: no root-clause edges} and~\ref{lemma: traversing a variable},
 to define a variable assignment.

\begin{lemma}\label{lemma: no root-clause edges}
Tree $T$ does not contain any root-clause edges.
\end{lemma}

\begin{proof}
First we observe that every root-clause edge has weight $1:K-1$.
Next we observe that for every root-clause edge $(u,v)$, be it an edge between the root and a 2N-clause or 2P-clause, there are (at least) three
edge disjoint paths between the vertices $u$ and $v$. Thus, if the edge $(u,v)$ appears in $T$, it has congestion at least $K-1+2=K+1$,
contradicting $\cng{G}{T} \le K$.
We conclude that no root-clause edge appears in $T$.
%
\end{proof}

To state Lemma~\ref{lemma: traversing a variable}, we need a few more definitions.
For a clause $\kappa$ containing variable $x$, a \emph{traversal from $\kappa$ to $x$} is a path in $T$ that starts from the center of $\kappa$ and  
contains exactly two external edges: the first being the clause-variable edge from $\kappa$ to $x$ (the \emph{entering edge}), 
and the second being any other external edge adjacent to the $x$-gadget (the \emph{exiting edge}). 
We say that \emph{$T$ traverses $x$ from $\kappa$} when such a path exists. Intuitively, a traversal from $\kappa$ to $x$ is a natural way to 
lift the edge $(\kappa,x)$ from $G^*$ to a path in $G$. We emphasize that a traversal begins at the \emph{centers} of flowers, 
as opposed to terminals or dummy nodes; this becomes relevant later in Claim~\ref{claim: connected clauses}.



\begin{lemma}\label{lemma: traversing a variable}
If $T$ traverses $x$ from a negative clause, then $T$ does not traverse $x$ from any positive clause.
\end{lemma}

Assuming the lemma above holds, we define our satisfying assignment as follows: for each variable $x$, make it false if $T$ traverses it from a 
negative clause, otherwise make it true. Lemma~\ref{lemma: no root-clause edges} implies that $T$ includes for each clause $\kappa$ a 
traversal from $\kappa$ to some variable $x$ appearing in it. (In particular, a traversal appears as a prefix of the path in $T$ from
 $\kappa$'s center to the center of the root.) $x$ is false when $\kappa$ is negative by definition of the assignment, 
 and $x$ is true when $\kappa$ is positive by Lemma~\ref{lemma: traversing a variable}. Therefore all clauses are satisfied.

\begin{proof}

\smallskip

	Let $x$ be a variable and $\alpha, \beta, \gamma$ 
	be the 2N-clause, 3P-clause, and 2P-clause of $x$, respectively. Assume for contradiction that $T$ traverses $x$ from $\alpha$ and at 
	least one positive clause $\beta$ or $\gamma$. We show that this assumption implies that $\cng{G}{T} > K$.
	 First we prove the following claim.

\begin{claim}\label{claim: connected clauses}
For some clause $\pi \in \braced{\beta, \gamma}$, the path in $T$ from the center of $\alpha$ to $\pi$ 
contains exactly two external edges, the first being the clause-variable edge $\parend{t_\alpha^x, \xtwoN}$ from $\alpha$ to $x$, 
and the second being the clause-variable edge $\parend{\gamma, \xtwoP}$ or $\parend{\beta, \xthreeP}$ from $x$ to $\pi$.
\end{claim}

To justify Claim~\ref{claim: connected clauses}, let $P$ be a traversal from $\alpha$ to $x$. If the exiting edge of $P$ is $\parend{\xtwoP, \gamma}$ 
(resp.\ $\parend{\xthreeP, \beta}$), then $P$ is simply the path in $T$ from the center of $\alpha$ to $\gamma$ (resp.\ $\beta$), and the 
claim is satisfied for $\pi = \gamma$ (resp.\ $\beta$).
Otherwise, the exiting edge of $P$ is $\parend{\xroot, t^x_{\rootflower}}$, which of course implies that $\xtwoN$ and $\xroot$ are both in $P$.
Now let $\pi \in \braced{\beta, \gamma}$ be a clause where $T$ traverses $x$ from $\pi$, and let $P'$ be a corresponding traversal. 
Then in $P'$, the entering edge must be succeeded by an edge containing $v \in \braced{\xtwoN, \xroot}$, due to the way variable gadgets are labeled. 
This implies that $P$ and $P'$ overlap. The path from $\alpha$'s center to $\pi$ is then given by joining the sub-path in $P$ from 
the center of $\alpha$ to $v$, followed by the sub-path in $P'$ from $v$ to $\pi$. Claim~\ref{claim: connected clauses} then follows.
	
\smallskip

Continuing the proof of Lemma~\ref{lemma: traversing a variable}, let $\pi \in \braced{\beta, \gamma}$ be a clause satisfying Claim~\ref{claim: connected clauses}. 
By Lemma~\ref{lemma: no root-clause edges}, the path in $T$ from the center of $\alpha$ to the center of the root must contain some root-variable edge; 
suppose $e=\parend{\yroot, t_\rootflower^y}$ is the first such edge on this path. Clearly, the center of the root lies on the shore 
$\treecutout{T}{\yroot}{e}$ of the cut $\cutof{\sptreecut{T}{e}}$, while the center of $\alpha$ lies on the opposite shore $\treecutin{T}{\yroot}{e}$. 
On the other hand, the path described in 
Claim~\ref{claim: connected clauses} does not contain any root-variable edges, implying that $\pi$ and $\xtwoN$ also lie on shore $\treecutin{T}{\yroot}{e}$. 
In conjunction with Observation~\ref{obs: integrity of core}, we obtain:
	 
\begin{corollary}
The core of the root is on the shore $\treecutout{T}{\yroot}{e}$, while $\xtwoN$, $\pi$, and the core of $\alpha$ are on the shore $\treecutin{T}{\yroot}{e}$.
\end{corollary}
	 
We now show that $\cng{T}{e} > K$.  The definition of $\problemMPNSAT$ implies that there are four distinct variables $z_1, z_2, z_3, z_4$, 
none equal to $x$, such that $z_1 \in \alpha$, $z_2,z_3 \in \beta$, and $z_4 \in \gamma$. This in turn implies that there are 7 edge-disjoint paths 
crossing $\cutof{\sptreecut{T}{e}}$. In particular, each path begins at either $\xtwoN$, $\pi$, or a core vertex of $\alpha$, 
then ends at a core vertex of the root. As explained in the discussion of flower gadgets, it is sufficient to specify terminals as path endpoints.
Three of these paths are 
\begin{alignat*}{5}
&
t_\alpha^{\rootflower, 1} - t_\rootflower^{\alpha, 1}, 
&&
\quad\quad\quad
&&
t_\alpha^{\rootflower, 2} - t_\rootflower^{\alpha, 2}, 
&&
\quad\quad\quad
&&
t_\alpha^{z_1} - \ztwoN_1 - \ztwoP_1 - \zroot_1 - t_\rootflower^{z_1}.
\end{alignat*}
The choice of the four remaining paths depends on whether $\pi = \beta$ or $\gamma$.
If $\pi=\beta$, these paths are
\begin{alignat*}{3}
&
\xtwoN - \xthreeP - \xroot - t_\rootflower^x, 
&&
\quad\quad\quad
&&
\beta -  \zthreeP_2 - \zroot_2 - t_\rootflower^{z_2},
\\
&\beta - \zthreeP_3 - \zroot_3 - t_\rootflower^{z_3},  
&&
\quad\quad\quad
&&
\xtwoN - \xtwoP - \gamma - t_\rootflower^\gamma.
\end{alignat*}
If $\pi=\gamma$, these paths are
\begin{alignat*}{3}
&
\xtwoN - \xtwoP - \xroot - t_\rootflower^x, 
&&
\quad\quad\quad
&&
\gamma - t_\rootflower^\gamma, 
\\
&
\gamma - \ztwoP_4 - \zroot_4 - t_\rootflower^{z_4},
&&
\quad\quad\quad
&&
\xtwoN - \xthreeP - \beta -  \zthreeP_2 - \zroot_2 - t_\rootflower^{z_2}.
\end{alignat*}
At most one of these paths contains $e$, implying $\cng{G,T}{e} \ge K-5 + 6 = K+1$, 
contradicting $\maxcng{G,T}\le K$. This completes the proof of Lemma~\ref{lemma: traversing a variable}.
\end{proof}


\section{Cactus Representation and Spanning-Tree Congestion}
\label{sec: cactus representation and congestion}

In this section, laying the groundwork for our algorithm in Section~\ref{sec: linear-time algorithm},
we analyze the structure of $K$-edge-connected graphs whose spanning-tree congestion is $K$.
We start, in Section~\ref{subsec: cactus representation},
by reviewing the properties of $K$-edge-connected graphs, captured by so-called \emph{cactus representation}.
Then, in Section~\ref{subsec: cactus for congestion K}, we focus on the special 
case of graphs with congestion $K$, and we prove that this congestion assumption implies 
additional structural properties of such graphs, that will lead to an efficient algorithm.

Throughout this section we assume that $G = (V,E)$ is a given $K$-edge-connected graph 
with $n = |V|$ vertices  and $m = |E|$ edges.


\subsection{Cactus Representation}
\label{subsec: cactus representation}

Two cuts $\cutof{X}$ and $\cutof{Y}$ are called \emph{nested} if one
of the four pair-wise shore intersections $X\cap Y$, $\barX\cap Y$,
$X\cap \bar Y$ and $\barX\cap \barY$ is empty.  If all four
intersections are non-empty then we say that cuts $\cutof{X}$ and
$\cutof{Y}$ \emph{cross}.  A family of cuts is called \emph{laminar}
if any two cuts in it are nested. By $\cutofin{X}{Y}$ we denote the subset of cut $\cutof{X}$ consisting of edges whose both endpoints are in $Y$.


\begin{theorem}\label{thm: k-cuts in k-edge-connected graphs}
{{\normalfont \cite{dinitz_etal_strukture_systemy_1976}}}.
Let $G$ be a $K$-edge-connected (multi-)graph. 
\begin{description}[nosep]
\item {\normalfont{(a)}}
If $K$ is odd, then $G$ has no crossing $K$-cuts.
That is, the family of $K$-cuts is laminar.
\item{\normalfont{(b)}}
If $K$ is even, then any two crossing $K$-cuts
$\cutof{X}$, $\cutof{Y}$ in $G$ satisfy 
$
|\cutofin{X}{Y}| 
	= |\cutofin{X}{\barY}|
	= |\cutofin{\barX}{Y}| 
	= |\cutofin{\barX}{\barY}| = K/2
$.
There are no edges between $X\cap Y$ and $\barX\cap\barY$,
and between $X\cap\barY$ and $\barX\cap Y$.
\end{description}
\end{theorem}


\begin{figure}[ht]
	\begin{center}
		\includegraphics[width = 4.5in]{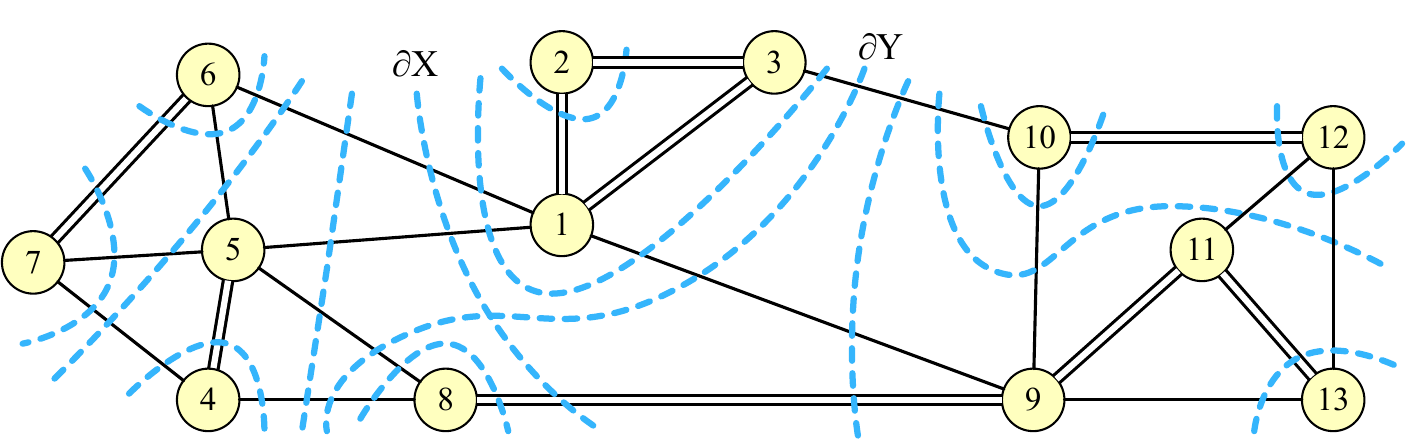}
        \end{center}
	\caption{A $4$-edge-connected multigraph and all its
          $4$-cuts. Parallel lines represent double parallel edges.
		Each vertex $v_i$ is identified by its index $i$.
		For $X = \braced{ v_{4},v_{5},v_{6},v_{7},v_{8} }$
		and $Y = \braced{v_{1},v_{2},v_{3},v_{4},v_{5},v_{6},v_{7}}$,
		cuts $\cutof{X}$, $\cutof{Y}$ cross.
		}
	\label{fig: 4-cuts_example_graph}
\end{figure}


This theorem can be refined to produce an even more informative
representation of $K$-cuts, in terms of so-called \emph{cactus graphs}.
For cactus graphs, we will use
terminology of \emph{nodes} and \emph{links} (instead of vertices and
edges).  A connected multigraph is called a \emph{cactus graph} if every link
belongs to exactly one cycle. (Equivalently, it is a $2$-edge-connected
graph whose biconnected components are cycles.)
Cactus cycles of length $2$ are called \emph{trivial}. 
A degree-$2$ node of a cactus is called an
\emph{external node}, and any other node is called an
\emph{internal node}. From the definition, it follows that each
cactus has at least one external node. If all its cycles are trivial, the cactus forms a tree whose
adjacent nodes are connected by two parallel links and the external nodes are its leaves.


\begin{theorem}\label{thm: k-cuts cactus representation}
{{\normalfont \cite{dinitz_etal_strukture_systemy_1976}}}.
Let $G = (V,E)$ be a $K$-edge-connected graph. Then
there is a cactus graph $\cactusC_G = (U,F)$ and an associated mapping $\phi : V \to U$ with the following properties:
\begin{description}[nosep]
\item {\normalfont{(a)}}
For each set $X\subseteq V$,
$\cutof{X}$ is a $K$-cut if and only if $X = \phi^{-1}(Q)$ for some $Q\subseteq U$ such that
$\cutof{Q}$ is a  $2$-cut of $\cactusC_G$. 
\item{\normalfont{(b)}}
If $K$ is odd, then all cycles in $\cactusC_G$ are trivial;
that is, $\cactusC_G$ is a tree with adjacent nodes connected by two parallel links.
\end{description}
\end{theorem}
The pair $\cactusC_G, \phi$ is called a \emph{cactus representation}
of $G$ (see Figures~\ref{fig: 4-cuts_example_graph} and~\ref{fig: 4-cuts_example_cactus}).
A cactus representation of $G$ of size $O(n)$ can be computed in 
near-linear time $\tildeO(m)$~\cite{karger_panigrahi_near-time-time_cactus_2009}.


\begin{figure}[ht]
	\begin{center}
		\includegraphics[width = 3in]{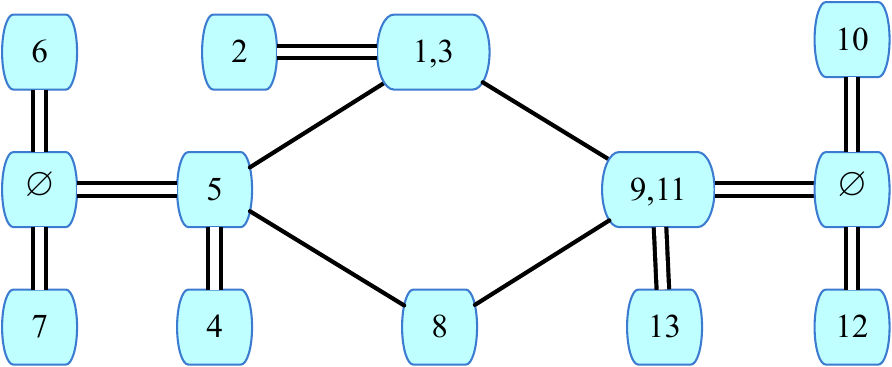}
        \end{center}
	\caption{
A cactus representation of the graph from Figure~\ref{fig: 4-cuts_example_graph}. The numbers represent indices of vertices in $G$
that are mapped by $\phi$ to the corresponding node in $\cactusC_G$.
Note that the pre-image of a node in $\cactusC_G$ could be empty. 
	}
	\label{fig: 4-cuts_example_cactus}
\end{figure}


\myparagraph{Basic $K$-cuts.}
Per Theorem~\ref{thm: k-cuts cactus representation}, each $K$-cut of
$G$ is represented by a $2$-cut of $\cactusC_G$.
Each $2$-cut of $\cactusC_G$ consists of two links that belong to the
same cycle. Thus, a cycle of length $\ell$ in $\cactusC_G$ represents
$\ell\choose 2$ $K$-cuts of $G$. Two $K$-cuts cross in $G$ if and only
if they are represented by two crossing $2$-cuts of $\cactusC_G$; that
is, the four links in these two $2$-cuts belong to the same cycle and
alternate in the order around this cycle.

A shore of any $K$-cut represented 
by two non-adjacent links of this cycle can be obtained as a union of
shores of $K$-cuts represented by pairs of its consecutive links. For our purpose it
is sufficient for us to focus on this subset of $\ell$ cuts represented by such link pairs.
This motivates the following definition.

A $K$-cut $\cutof{X}$ in $G$ is called a \emph{basic $K$-cut} if it is
represented by a pair of links of $\cactusC_G$ that share a node.  If
this shared node is $b$, we say that $\cutof{X}$ is a \emph{basic
$K$-cut associated with $b$}.  Note that $K$-cuts represented by two
parallel links are also basic and are associated with both endpoints of these links. 

From now on, as a rule, when talking about a cut
$\cutof{X}$ associated with $b$, we will represent it by
$X=\phi^{-1}(Q)$ for the shore $Q$ in the $2$-cut of $\cactusC_G$ that
does \emph{not} contain $b$.

For each node $b\in\cactusC_G$ of degree $2d$, the links
incident with $b$ form $d$ disjoint pairs with each pair on the same
cycle of $\cactusC_G$. All the basic $K$-cuts associated with $b$ are
then given by these pairs. Representing them by $\cutof{Z_1}$,
\ldots, $\cutof{Z_d}$ as described above, all the sets
$Z_1$, \ldots, $Z_d$ form a disjoint partition of the 
set $V\setminus\phi^{-1}(b)$.

Basic $K$-cuts form a laminar family. In fact, Theorem~\ref{thm: k-cuts in k-edge-connected graphs}
implies even the following two
stronger properties that will be crucial for our algorithm.

\begin{observation}\label{obs:basic-laminar}
(Non-crossing property of basic cuts)
\begin{description}[nosep]
\item {\normalfont{(a)}}
  A basic $K$-cut does not cross any other $K$-cut (even a non-basic one).
  \label{obs: cuts vs cuts Xj}
\item {\normalfont{(b)}}
Let $\cutof{Z_1}$, \ldots, $\cutof{Z_d}$ be all basic $K$-cuts
associated with $b\in\cactusC_G$. Then for any $K$-cut
$\cutof{Y}$, there exists $j$ for which either $Y\subseteq Z_j$ or
$\barY\subseteq Z_j$.
\end{description}
\end{observation}


\subsection{Cactus Representation for Graphs with Congestion $K$}
\label{subsec: cactus for congestion K}


We now assume that $\stc{G}\le K$, and we show that this assumption implies
additional properties of $G$'s cactus representation. These properties will
play a critical role in our algorithm.

By $\Tstar$ we denote a spanning tree of $G$ with $\cng{G}{\Tstar}\le K$.
Note that by the assumption about $K$-edge-connectivity, $|E_u|\geq K$ for each vertex $u$, and $\stc{G}=K$. 
So each edge $e$ of $\Tstar$ induces a $K$-cut, i.e., we have $|\cutof{\sptreecut{\Tstar}{e}}| = K$.
In particular, every leaf $u$ of $\Tstar$ has degree exactly $K$ in $G$ and $E_u=\cutof{\braced{u}}$ is a (trivial) $K$-cut. 


\begin{observation}\label{obs: cactus nodes singletons or empty}
For every node $b\in\cactusC_G$, we have $|\phi^{-1}(b)|\le 1$.
\end{observation}

\begin{proof}
Suppose that $\phi^{-1}(b)$ contains two distinct
vertices $u$ and $v$.  Then none of the $K$-cuts represented by
$\cactusC_G$ separates $u$ from $v$.  On the other hand, any edge on
the $u$-to-$v$ path in $\Tstar$ induces a $K$-cut separating $u$
and $v$, a contradiction.
\end{proof}

By Observation~\ref{obs: cactus nodes singletons or empty},
each node $b$ of $\cactusC_G$ can be classified into one of the two categories:
either $\phi^{-1}(b) = \emptyset$, in which case we refer to $b$ as a
\emph{Type-$0$ node}, or $|\phi^{-1}(b)| = 1$, in which case we call
it a \emph{Type-$1$ node}. This is refined further in
the observation below that follows directly from
the fact that if $\deg(u)=K$ then $E_u=\cutof{\braced{u}}$ is a
trivial $K$-cut and $\phi(u)$ is a Type-1 external node of $\cactusC_G$.


\begin{observation}\label{obs: odd K singleton node large degree}
Every node $b\in\cactusC_G$ is of one of the following three types:
\begin{description}[nosep]
\item{\normalfont{(i)}} 
an external Type-$1$ node with $b=\phi(u)$ for a vertex $u$
    of degree exactly $K$,
\item{\normalfont{(ii)}}
an internal Type-$1$ node with $b=\phi(u)$ for a vertex $u$
    of degree strictly greater than $K$, or
\item{\normalfont{(iii)}}
an internal Type-$0$ node.
\end{description}
\end{observation}


\myparagraph{Basic $K$-cuts associated with cactus nodes.}  
A key property of $\Tstar$ that we use in our algorithm is that for any
basic $K$-cut $\cutof{X}$, all edges in $\Tstar\cap\cutof{X}$ share an
endpoint $w$. Lemmas~\ref{lem: K-cuts, type 1 node} and~\ref{lem:
  K-cuts, type 0 node} below show an even stronger property, namely that
for all basic cuts associated with any given node $b$, this common
endpoint will be the same; furthermore for Type-1 node $b$, we have
$\phi(w)=b$, i.e., we have only this single choice for the common
endpoint $w$. See Figure~\ref{fig: cactus nodes and its cuts} for an
illustration.


\begin{figure}[ht]
	\begin{center}
		\includegraphics[width = 1.75in]{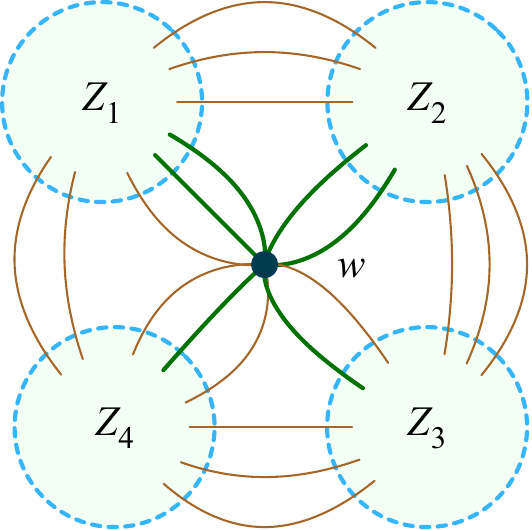}
		\hspace{0.5in}
		\includegraphics[width = 1.75in]{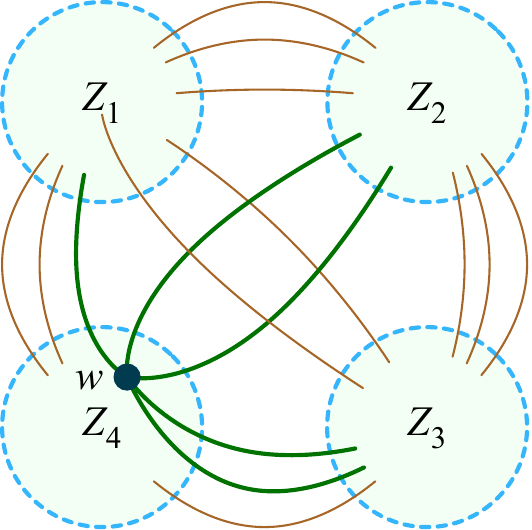}
	\end{center}
	\caption{On the left, an illustration of Lemma~\ref{lem: K-cuts, type 1 node}.
	On the right, an illustration of Lemma~\ref{lem: K-cuts, type 0 node}.
	In both examples $K=8$ and $d=4$. Edges in $\Tstar$ that
	cross the basic $K$-cuts  $\cutof{Z_i}$ are green (dark) and thick, non-tree edges are brown (light) and thin.}
	\label{fig: cactus nodes and its cuts}
\end{figure}


\begin{lemma}\label{lem: K-cuts, type 1 node}
Let $b$ be a Type-$1$ node and $w\in V$ such that $\phi(w)=b$.
Let $\cutof{Z_1}$, \ldots, $\cutof{Z_d}$ be all basic $K$-cuts associated with $b$.
Then $\Tstar \cap \bigcup_{i=1}^d \cutof{Z_i} \;\subseteq\; E_w$.
\end{lemma}

\begin{proof}
It is sufficient to prove that for any $i\neq i'$, there is no edge in
$\Tstar$ connecting $Z_i$ and $Z_{i'}$. For a contradiction, assume
that $u\in Z_i$, $v\in Z_{i'}$, and $(u,v)\in\Tstar$. Now consider the
first edge $e$ on the path from $w$ to $u$ in the spanning tree
$\Tstar$. Edge $e$ induces a $K$-cut $\cutof{Y}=\cutof{\Tstar_e}$
with has both $u$ and $v$ in the same shore and $w$ in the other
shore. Thus for any $j$, none of the shores of $\cutof{Y}$ can be
contained in some set $Z_j$. This contradicts Observation~\ref{obs: cuts vs cuts Xj}.
\end{proof}


\begin{lemma}\label{lem: K-cuts, type 0 node}
Let $b$ be a Type-$0$ node, and let $\cutof{Z_1}$, \ldots, $\cutof{Z_d}$ be all basic
$K$-cuts associated with $b$. Then there exists $j$ and $w\in Z_j$ such that 
$\Tstar \cap \bigcup_{i=1}^d \cutof{Z_i} \;\subseteq\; E_w$.
\end{lemma}

\begin{proof}
Suppose, towards contradiction, that the claim in the lemma is false. That is, there
are edges $(u,v), (u',v') \in \Tstar \cap \bigcup_{i=1}^d \cutof{Z_i}$ with all endpoints $u,v,u',v'$ distinct.
Consider the unique path in $\Tstar$ from $\braced{u,v}$ to $\braced{u',v'}$ that does not include edges $(u,v)$ and $(u',v')$. 
Let $e$ be any edge on this path and let $\cutof{Y} =
\cutof{\sptreecut{\Tstar}{e}}$ be the $K$-cut induced by $e$.
Then $\cutof{Y}$ separates $\braced{u,v}$ from $\braced{u',v'}$. However, $u$
and $v$ are in two different sets $Z_i$; the same holds for $u'$ and
$v'$. Thus for any $j$, none of the shores of $\cutof{Y}$ can be
contained in some set $Z_j$. This contradicts Observation~\ref{obs: cuts vs cuts Xj}.
\end{proof}

The corollary below will play a key role in our algorithm, and it
follows directly from Lemmas~\ref{lem: K-cuts, type 1 node} and~\ref{lem: K-cuts, type 0 node}, as each
basic cut is associated with some node $b$. 

\begin{corollary}\label{cor: odd K, cross edges common endpoint}
For any basic $K$-cut $\cutof{Z}$ of $G$, all edges in
$\Tstar\cap\cutof{Z}$ have a common endpoint.
\end{corollary}


\myparagraph{Basic $K$-cuts of cactus cycles.}
Assume now that $K$ is even, and let $C = a_0 - a_1 - \ldots - a_{\ell-1}
- a_0$ be a non-trivial cycle in $\cactusC_G$, so $\ell\geq 3$.  We
index $C$ cyclically, i.e., $a_i = a_{i \pmod \ell}$, for all integers $i$.
For each $i$, the two consecutive links $(a_{i-1},a_i)$ and $(a_i,a_{i+1})$ of $C$
form a $2$-cut $\cutof{Q_i}$ of $\cactusC_G$, where the shore $Q_i$ is chosen so that $a_i\in Q_i$.
Let $\cutof{Z_i}$ be the basic $K$-cut in $G$ represented by the $2$-cut 
$\cutof{Q_i}$, with $Z_i=\phi^{-1}(Q_i)$.
The sets $Z_0,...,Z_{\ell-1}$ form a partition of $V$, that is $Z_{i}\cap Z_{i'} = \emptyset$ for $i\neq i'$
and $\bigcup_{i=0}^{\ell-1} Z_i = V$.
By Theorems~\ref{thm: k-cuts in k-edge-connected graphs} and~\ref{thm: k-cuts cactus representation}, 
for any $i$, there are exactly $K/2$
edges between $Z_i$ and $Z_{i+1}$, i.e., 
$|\cutof{Z_i}\cap\cutof{Z_{i+1}}|=K/2$;  we call this set
$\cutof{Z_i}\cap\cutof{Z_{i+1}}$ of $K/2$ edges a \emph{half-$K$-cut}. It follows that
$\cutof{Z_i}\cap\cutof{Z_{i'}}=\emptyset$ if $a_i$ and $a_{i'}$ are not consecutive on $C$, 
i.e., whenever $i-i' \notin \braced{1,-1} \pmod{\ell}$.

\begin{figure}[ht]
	\begin{center}
		\includegraphics[width = 3.5in]{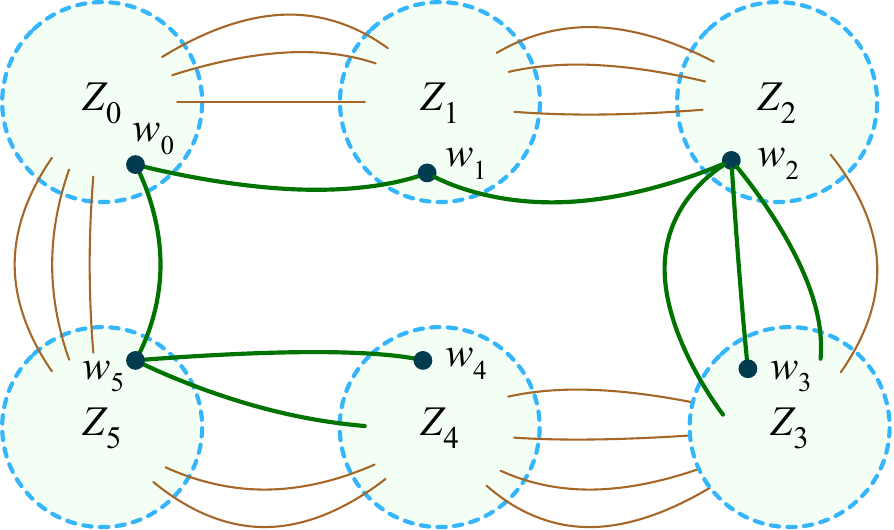}
	\end{center}
	\caption{An illustration of Lemma~\ref{lem: K-cuts, cycle}, for $K=8$, $\ell = 6$ and $g=4$.}
	\label{fig: cactus cycle and its cuts}
\end{figure}

Next, in the lemma below, we show that $\Tstar$ traverses the cuts represented by $C$ in a very restricted way:
for some index $g$, the edges of $\Tstar$ form a length-$(\ell-1)$ path that traverses
consecutive half-$K$-cuts  between sets $Z_{g}, Z_{g+1},...,Z_{g+\ell-1} = Z_{g-1}$,
and the only other edges of $\Tstar$ in the cuts of $C$ are among those connecting a vertex $w_{g+1} \in Z_{g+1}$ to set $Z_g$,
and among those connecting a vertex $w_{g-2}\in Z_{g-2}$ to set $Z_{g-1}$. In particular, $\Tstar$ has
no edges in the half-cut $\cutof{Z_{g-1}}\cap\cutof{Z_g}$.
(See Figure~\ref{fig: cactus cycle and its cuts}.)


\begin{lemma}\label{lem: K-cuts, cycle}
There is an index $g$ and vertices $w_g, w_{g+1}, \ldots, w_{g+\ell-1} = w_{g-1}$, with $w_i\in Z_i$ for $i = g,\ldots,g+\ell-1$,
such that the edge set $\Tstar \cap \bigcup_{i=0}^{\ell-1}\cutof{Z_i}$ has the following form:
\begin{description}[nosep]
\item{\normalfont{(i)}} 
$\Tstar\cap \cutof{Z_i}\cap \cutof{Z_{i+1}} = \braced{(w_{i},w_{i+1})}$ for $i = g+1,\ldots,g+\ell-3$,
\item{\normalfont{(ii)}} 
$(w_g,w_{g+1}) \in \Tstar \cap \cutof{Z_{g}} \cap \cutof{Z_{g+1}} \subseteq E_{w_{g+1}}$, 
	$(w_{g-2},w_{g-1}) \in \Tstar\cap \cutof{Z_{g-1}}\cap \cutof{Z_{g-2}} \subseteq E_{w_{g-2}}$, and
\item{\normalfont{(iii)}} 
	$\Tstar\cap \cutof{Z_{g}}\cap \cutof{Z_{g-1}} = \emptyset$.
\end{description}
\end{lemma}

\begin{proof}
First observe that if $\Tstar$ contains two edges in adjacent
half-$K$-cuts, say $e_{i-1}\in \Tstar\cap\cutof{Z_{i-1}}\cap \cutof{Z_{i}}$
and $e_i\in\Tstar\cap\cutof{Z_{i}}\cap \cutof{Z_{i+1}}$, then, by Corollary~\ref{cor: odd K, cross edges common endpoint},
they have a common endpoint $w_i$. Further, since $Z_{i-1}\cap Z_{i+1} = \emptyset$ (and $\ell\ge 3$), we must have $w_i\in Z_i$.

It is not possible for each half-$K$-cut $\cutof{Z_{i}}\cap \cutof{Z_{i+1}}$ to have an edge from $\Tstar$, because then,
by the observation from the previous paragraph, these edges would form a cycle.
Choose one $g$ for which $\Tstar\cap \cutof{Z_{g}}\cap \cutof{Z_{g-1}} = \emptyset$. So~(iii) is already true.
Since $\Tstar$ is spanning, all other half-$K$-cuts intersect $\Tstar$. Choose one edge from
$\Tstar\cap\cutof{Z_{i}}\cap \cutof{Z_{i+1}}$ for each $i\neq g-1$. 
Using the observation from the first paragraph, these chosen edges form a path $w_g, \ldots, w_{g+\ell-1}$
with $w_i\in Z_i$ for $i = g,\ldots,g+\ell-1$, and this path satisfies properties~(i) and~(ii).
\end{proof}



\section{A Linear-Time Algorithm}
\label{sec: linear-time algorithm}

In this section we prove Theorem~\ref{thm: algorithm for k-edge-connected}, by developing
an $\tildeO(m)$-time algorithm that, given a $K$-edge-connected graph $G$, determines whether $\stc{G}=K$.
We start by sketching the basic ideas leading to our algorithm.

A natural attempt to design an algorithm would be to exploit the tree-like structure of cactus graphs,
applying the dynamic programming strategy to recursively compute some congestion-related information for the $K$-cuts of $G$.
The first challenge one encouters is that in general (for even values of $K$), the $K$-cuts do not form a laminar structure. 

This is where the concept of \emph{basic} $K$-cuts, introduced in Section~\ref{subsec: cactus representation}, is
helpful. The family of basic $K$-cuts \emph{is} laminar.
More specifically, we will choose an arbitrary degree-$K$ vertex $r$ in $G$, called the \emph{root} of $G$ 
(see Section~\ref{subsec: rooting G and cactus}.)
We can then identify each basic $K$-cut $\cutof{Z}$ by the shore $Z$ that does not contain $r$.
The subset relation between these shores defines a tree structure on the basic $K$-cuts.
The algorithm can then process the basic $K$-cuts bottom-up in this order.

What information such an algorithm would need to maintain for each basic $K$-cut $\cutof{Z}$?
A na\"{\i}ve approach would be to somehow keep track of all ``candidate crossing-edge sets'' $A\subseteq\cutof{Z}$,
namely the sets $A$ for which there is a spanning tree $\Tstar$ of $G$ with congestion $K$ such that $\Tstar\cap \cutof{Z} = A$.
This raises two issues. One, the resulting algorithm's running time would be exponential in $K$.
Two, obviously, the algorithm does not yet know $\Tstar$ when $\cutof{Z}$ is considered.
What's worse, even the congestion of tree edges that are wholly inside shore $Z$ cannot be uniquely determined 
based only on the subgraph induced by $Z$. (It may be possible to address the latter issue by storing appropriate
information about the tree topology inside $Z$, but at the price of further increasing the time complexity.)

The characterization of $K$-edge-connected graphs with congestion $K$, 
developed in Section~\ref{sec: cactus representation and congestion}, suggests an alternative approach.
The key property is Corollary~\ref{cor: odd K, cross edges common endpoint}, 
which says that if $\Tstar$ is a spanning tree with congestion $K$ then
for each basic $K$-cut $\cutof{Z}$, all edges in $\Tstar\cap \cutof{Z}$ have a common endpoint.
Instead of focussing on tree edges crossing $\cutof{Z}$, we can instead 
attempt to recursively compute these common endpoints.
This is not quite possible, for the same reason as above: we cannot determine if a vertex
is such a common endpoint, for some congestion-$K$ spanning tree, without knowing the whole graph $G$.
Generally, the analysis in Section~\ref{sec: cactus representation and congestion}, while it
provides crucial ideas, is not sufficient in itself to derive a dynamic-programming algorithm
because it is expressed in terms of global properties of $G$.
For this, we need a definition of ``candidate common endpoints'' that is based only on the
information from the already processed subgraph.

To address this, in Section~\ref{subsec: rooting G and cactus} we establish analogs
of the properties in Section~\ref{sec: cactus representation and congestion} expressed in terms of
the rooted versions of $\cactusC_G$ and $G$.
One key observation is this: for any vertex $w$ and for \emph{any} spanning tree $T$
such that $T\cap \cutof{Z} \subseteq E_w$, the congestion of $T$'s edges within the subgraph induced
by $Z\cup\braced{w}$ depends only on this subgraph. If all these congestion values are $K$,
we call $T$ \emph{safe}, and we call $w$ a \emph{hub} for $\cutof{Z}$. (We also require, for technical reasons,
that $w$ has at least one edge crossing $\cutof{Z}$.) 
If $w$ is a hub for $\cutof{Z}$, it only means that the information from the subgraph $Z\cup\braced{w}$
is not sufficient to eliminate $w$ as a possible common endpoint of the edges from $\cutof{Z}$ that
belong to some congestion-$K$ spanning tree of $G$. (Later, when processing some ancestor $K$-cut of $Z$,
it may turn out that such a tree does not exist for $w$.)
We prove that the set of hubs of $\cutof{Z}$, denoted $\hubs(Z)$, can be
determined from the hub sets of $\cutof{Z}$'s children, thus establishing a recurrence relation that
drives the dynamic programming algorithm (see Section~\ref{subsec: the algorithm}).

While it may not be obvious, computing all hub sets $\hubs(Z)$ is sufficient to determine
whether the congestion of $G$ is $K$. (That is, for any $w\in \hubs(Z)$ we do not need to keep track of which tree
edges from $w$ cross $\cutof{Z}$, circumventing the issue we mentioned earlier.)
The reason is this: once (and if) we reach the root $r$,
the condition on $r$ being a hub for its cut $\cutof{Z}$ (with $Z = V\setminus\braced{r}$)
is equivalent to $G$ having a spanning tree with congestion $K$.


\subsection{Rooting $G$ and its Cactus Graph}
\label{subsec: rooting G and cactus}

Throughout this sub-section (Section~\ref{subsec: rooting G and cactus}) 
we will assume that $G$ is a $K$-edge-connected graph and that $|\phi(b)|\le 1$ for every node
$b\in\cactusC_G$, so that Observation~\ref{obs: odd K singleton node large degree} can be applied.
We fix an arbitrary degree-$K$ vertex $r\in V$ as the \emph{root of $G$}.  By convention, from now
on, each $K$-cut $\cutof{X}$ will be represented by the shore $X$ that
does not contain $r$. This naturally imposes a tree-like structure on
basic $K$-cuts, where a basic $K$-cut $\cutof{X}$ is a
\emph{descendant} of a basic $K$-cut $\cutof{Y}$ if $X\subsetneq Y$. A
descendant $\cutof{X}$ is called a \emph{child} of $\cutof{Y}$ (making
$\cutof{Y}$ a \emph{parent} of $\cutof{X}$) if there is no basic $K$-cut
$\cutof{Z}$ such that $X\subsetneq Z \subsetneq Y$.  Since the family
of basic $K$-cuts is laminar by Observation~\ref{obs:basic-laminar}, every cut except
$\cutof{(V\setminus\{r\})}$ has a unique parent. (So the descendant
relation is indeed a tree.)

To reflect this ordering of basic $K$-cuts in $G$'s cactus representation,
we root $\cactusC_G$ at the node $\phi(r)$. In our algorithm
we will need a few more related concepts:

  \begin{itemize}[nosep]
  \item A node $a$ of $\cactusC_G$ is said to be \emph{below} a node $b$ if $b$
    is on every path from $a$ to $\phi(r)$.
  \item A node $a$ is called the \emph{head} of the cycle $C$ in $\cactusC_G$, and denoted
    $a_C$, if $a\in C$ and all the other nodes of $C$ are below $a_C$.
  \item For a node $a\neq\phi(r)$, by $Q^a$ we denote the set of nodes consisting of $a$
    and all nodes below $a$. Also, let $W^a=\phi^{-1}(Q^a)$ be the set of
	vertices of $G$ represented by $Q^a$.
  \item For a cycle $C$ in $\cactusC_G$, 
  	we let $Q^C = \bigcup_{b\in C\setminus a_C}Q^b$. That is, $Q^C$ consists
	of all nodes that are below $a_C$, excluding $a_C$ itself. Also, let $W^C=\phi^{-1}(Q^C)$.
  \end{itemize}

The cactus structure implies that each cycle has a unique head, so the
definition above is valid. Furthermore, each node except for $\phi(r)$
belongs to a unique cycle, which could be trivial or not, where it is a non-head node.
On the other hand, a node can be a head of several cycles.

\smallskip

The following observation, that follows directly from the definitions,
summarizes the properties of the basic $K$-cuts and their ordering,
and the corresponding properties of $\cactusC_G$.


\begin{observation}\label{obs:cactus-correspondence}
The basic $K$-cuts of $G$ satisfy the following properties:
\begin{enumerate}[{\rm(a)},nosep]
\item
  The family of basic $K$-cuts of $G$ consists exactly of all $K$-cuts
  $\cutof{W^b}$ and $\cutof{W^C}$, for nodes $b\neq\phi(r)$ 
  and non-trivial cycles $C$ in $\cactusC_G$.
  Further, this representation is unique; that is all $K$-cuts $\cutof{W^b}$
  and $\cutof{W^C}$, for nodes $b\neq\phi(r)$ and non-trivial cycles
  $\cutof{W^C}$ in $\cactusC_G$, are different.
\item
  If $C$ is a trivial cycle of $\cactusC_G$
  then $\cutof{W^C} = \cutof{W^b}$, where $b$ is the node on $C$ different from $a_C$.
\item
  For nodes $a$ and $b$ of $\cactusC_G$ distinct from $\phi(r)$, 
  $a$ is below node $b$ if and only if $\cutof{W^a}$ is a descendant of $\cutof{W^b}$.
\item
  For a non-trivial cycle $C$ of $\cactusC_G$, the parent of $\cutof{W^C}$ is $\cutof{W^{a_C}}$, and
  the children of $\cutof{W^C}$ are the $K$-cuts $\cutof{W^b}$ for $b\in C\setminus\braced{a_C}$.
  Further, $W^C = \bigcup_{b\in C\setminus\braced{a_C}} W^b$.
\item
  For a node $b\neq \phi(r)$ of $\cactusC_G$, if $C$ is the unique cycle where $b$ is a non-head node,
  the parent of $\cutof{W^b}$ is either $W^{a_C}$, if $C$ is trivial, or $W^C$, if $C$ is non-trivial.
  The children of $W^b$ are all $K$-cuts $W^C$, for cycles $C$ for which $a_C = b$.
  (Note that, by~(b), for trivial cycles $C$ these children also have the form $W^a$, where
  $a$ is the node on $C$ other than $b$.)
  Further, $W^b = \phi^{-1}(b) \cup \bigcup_{C:b = a_C} W^C$.
\end{enumerate}
\end{observation}

Finally, note that $\phi(r)$ is an external node of $\cactusC_G$, as $r$ has degree $K$ in $G$. 
It follows that there is a unique cycle $C$ with $a_C = \phi(r)$. Thus
the $K$-cut $\cutof{(V\setminus\{r\})}$ is equal to $W^C$.
In particular, if $C$ is trivial, this $K$-cut is $W^b$, for the node $b$ in $C$ other than $\phi(r)$.


\subsubsection{Safe Trees and Hubs}
\label{sec: safe trees and hubs}

Through the rest of this section, by a \emph{tree in $G$} we will mean a tree that is a subgraph of $G$. 
To simplify notation we will sometimes treat trees in $G$ as sets of vertices
and for such a tree $T$ we will write $\cutof{T}$ to mean $\cutof{V(T)}$.

We now define the concepts of \emph{safe trees} and \emph{hubs}, mentioned earlier in 
the beginning of Section~\ref{sec: linear-time algorithm}. 
Roughly, a tree in $G$ is considered \emph{safe} if it cannot be eliminated as a possible
subtree of $\Tstar$ (a spanning tree with congestion $K$) based only on the subgraph of $G$ it spans.
These trees are used as partial solutions constructed (implicitly) in the algorithm.  
As the algorithm proceeds, it may expand a safe tree, but only by
connecting its root to a vertex in the new, larger safe tree. 
A safe tree can get discarded, if we discover that it cannot be expanded to a larger safe tree.

To formally define safe trees, we need to extend some concepts from Section~\ref{sec: preliminaries},
defined for spanning trees, to arbitrary trees in $G$. If $T$ is any tree in $G$
and $e=(u,v)\in T$ is an edge of $T$, removing $e$ from $T$
disconnects $T$ into two connected components. As before, given a vertex $w$ of
$T$, we denote the component not containing $w$ by $\treecutout{T}{w}{e}$. 
As usual, by $\cutof{\treecutout{T}{w}{e}}$ we denote the set of edges between 
this component and the rest of the graph.

For any tree $T$ in $G$ and any vertex $w$ in $T$,
$(T,w)$ denotes the rooted version of $T$, with $w$ designated as its
root. Such a rooted tree $(T,w)$ is called \emph{safe} if for each edge $e$ of $T$ we have
$|\cutof{\treecutout{T}{w}{e}}|=K$.
The concept of a safe tree is a natural generalization of the
notion of a spanning tree with congestion $K$, in the following sense:


\begin{observation}
\label{obs:hub-final}
Let $T$ be a spanning tree in $G$ and $w$ any vertex. Then $(T,w)$ is
safe if and only if $\maxcng{G,T}\le K$. 
\end{observation}

The lemma below generalizes Corollary~\ref{cor: odd K, cross edges common endpoint} to an arbitrary safe tree.


\begin{lemma}
\label{lem:safe-tree-common-endpoint}
Let $\cutof{Z}$ be a basic $K$-cut and $(T,w)$ a safe tree in $G$ with $T\cap\cutof{Z}\neq\emptyset$.
Then all edges of $T\cap\cutof{Z}$ have a common endpoint.
\end{lemma}

\begin{proof}
For a contradiction, suppose that $f$ and $f'$ are two edges in
$T\cap\cutof{Z}$ such that all four of their endpoints are
distinct. Consider a path connecting these edges in $T$ and any edge
$e$ on this path in-between $f$ and $f'$. Then the cut
$\cutof{\treecutout{T}{w}{e}}$ is a $K$-cut, by the definition of a safe tree. 
However, $\treecutout{T}{w}{e}$ contains exactly one of the
edges $f$ and $f'$ with its both endpoints, while the endpoints of the other edge are
outside of $\treecutout{T}{w}{e}$. It follows that the cuts
$\cutof{Z}$ and $\cutof{\treecutout{T}{w}{e}}$ are crossing, as each
of the four endpoints is in a different shore intersection. 
Since $Z$ is a basic $K$-cut, this is a contradiction with Observation~\ref{obs:basic-laminar}.

From the above paragraph, any pair of edges in $T\cap\cutof{Z}$ must share
an endpoint. We still need to argue that \emph{all} edges in $T\cap\cutof{Z}$
must share a common endpoint. Indeed, this is trivial when $|T\cap\cutof{Z}|\in\braced{1, 2}$.
When $|T\cap\cutof{Z}| \ge 3$, fix any three edges in $T\cap\cutof{Z}$, and observe that they
must share a common endpoint, say $w$, because otherwise they would form a $3$-cycle.
Any other edge in $T\cap\cutof{Z}$, to share an endpoint with each of these three edges, must have $w$ as an endpoint.
\end{proof}

We are now ready to define hubs.
Let $\cutof{Z}$ be a basic cut in $G$. A node $w\in V(\cutof{Z})$ is
called a \emph{hub} for $Z$, if there exists a spanning tree
$T$ of the subgraph of $G$ induced by $Z\cup\{w\}$ such that $(T,w)$ is a safe tree in $G$.  
We call this $T$ a \emph{witness tree} for $w$ and $\cutof{Z}$.  If $w\in Z$,
then $w$ is called an \emph{in-hub for $Z$}, otherwise $w$ is an \emph{out-hub for $Z$}.
By $\hubs(Z)$ we denote the set of all hubs for $Z$.


\subsubsection{Constructing the Hub Sets for Basic $K$-Cuts}
\label{sec:constructing-hubs}

We now develop the properties of hubs for basic $K$-cuts $\cutof{W^b}$ that are
analogous to Lemmas~\ref{lem: K-cuts, type 1 node} and~\ref{lem: K-cuts, type 0 node}, but
are formulated in terms of safe trees instead of the congestion-$K$ spanning tree $\Tstar$.

The following observations follow directly from the definitions. They show that the safety
property of trees is preserved under some operations, like taking sub-trees or combining
disjoint trees with a common root.

\begin{observation}
\label{obs:hub-subtree}
Let $(T,w)$ be a safe tree, $u$ a vertex of $T$, and $e$ the first edge on the
path in $T$ from $u$ to $w$. Then $(\treecutout{T}{w}{e},u)$ is also a safe tree.
\end{observation}

\begin{observation}
\label{obs:hub-union}
Let $(T,w)$ and $(T',w)$ be two edge-disjoint safe trees with a common
root $w$. Then $(T\cup T',w)$ is a safe tree.
\end{observation}

\begin{observation}
\label{obs:hub-edge}
Let $(T,v)$ be a rooted tree in $G$ satisfying $|\cutof{T}|=K$, and let $(u,v) \in \cutof{T}$.
Then $(T\cup\{(u,v)\},u)$ is a safe tree if and only if $(T,v)$ is a
safe tree.
\end{observation}

Let $Z = W^b$ be a basic $K$-cut, for some node $b$ of $\cactusC_G$.
We now describe how the set of hubs $\hubs(Z)$ can be computed from the 
sets of hubs for the children of cut $\cutof{Z}$. 
We break it into three cases: when $b$ is an external node,
when it is an internal node of Type~1, and when it is an internal node of Type~0.

\smallskip
The case when $b$ is an external node (so it has no children) is simple. 
Recall that each external node is a Type-$1$ node with a degree $K$ vertex in its preimage. 
The set $W^b$ then contains this single vertex and represents a trivial $K$-cut.


\begin{lemma}
\label{lem:hub-leaf}
Let $b\neq\phi(r)$ be an external node of $\cactusC_G$ (necessarily of Type-$1$)
with $\phi^{-1}(b) = \braced{v}$. Then $\hubs(W^b)= \{v\} \cup N(v)$.
\end{lemma}

\begin{proof}
The witness tree for $v$ is the trivial tree with $v$ as a single
vertex. For any $w \in N(v)$, the witness tree is the tree with a
single edge $(v,w)$. There are no other vertices in $V(\cutof{W^b})$, so the lemma follows.
\end{proof}

\smallskip
If $b$ is an internal node, of any type, we will denote the
children of $W^b$ by $Z_1,\ldots,Z_t$. These children are
determined as detailed in Observation~\ref{obs:cactus-correspondence}(e).
As also explained in Observation~\ref{obs:cactus-correspondence}(e),
the set $W^b$ is a disjoint union of all sets $Z_i$, plus a singleton $\phi^{-1}(b)$ if $b$ is of Type-$1$. 
This, together with Observation~\ref{obs:basic-laminar}, implies the following key property: 

\begin{observation}\label{obs:basic-laminar rooted}
Let $b$ be an internal node of $\cactusC_G$, and denote by $Z_1,...,Z_t$ the children of
its corresponding $K$-cut $W^b$. 
If $\cutof{X}$ is a $K$-cut (not necessarily basic) such that $X\subsetneq W^b$
then $X\subseteq Z_i$ for some $i$. 
\end{observation}

\smallskip
Now we examine the case when $b$ is an internal node of Type-$1$.
We show that the only possible in-hub for $W^b$ is the vertex $v$ with $\phi(v)=b$, and that for $v$
to be a hub it must be an out-hub for each $Z_i$. The only possible out-hubs
are neighbors of $v$ outside $W^b$, providing that $v$ is an in-hub itself. 
Formally, we have the following lemma.


\begin{lemma}
\label{lem:hub-type1}
Let $b$ be a Type-$1$ internal node of $\cactusC_G$, and let $v$ be the (unique)
vertex of $G$ satisfying $\phi(v)=b$. Let $Z_1,\ldots,Z_t$ be the children of $W^b$.
Then $v\in V(\cutof{W^b}) \cap \bigcap_{i=1}^t V(\cutof{Z_i})$ and
\begin{equation}
\hubs(W^b) \;=\; \begin{cases*}
		\braced{v}\cup(N(v)\setminus W^b) & if $v \in \bigcap_{i=1}^t \hubs(Z_i)$
		\\
		\emptyset & otherwise
		\end{cases*}	
		\label{eqn: hub formula type 1}				
\end{equation}
\end{lemma}

\begin{proof}
The first condition, namely that $v$ has edges crossing $K$-cut $\cutof{W^b}$ and all $K$-cuts $\cutof{Z_i}$,
follows directly from Lemma~\ref{lem: K-cuts, type 1 node}. 

In the rest of the proof we show that Equation~(\ref{eqn: hub formula type 1}) is true. The argument
is by considering three types of vertices: vertices in $\bigcup_{i=1}^t Z_i$,
vertex $v$ (the only vertex in $W^b \setminus \bigcup_{i=1}^t Z_i$), and vertices in $V\setminus W^b$.
The theorem will follow from the three claims established below.


\begin{claim}\label{cla: hub-type1, w in Z's}
$\hubs(W^b) \cap (\bigcup_{i=1}^t Z_i) = \emptyset$.
\end{claim}

To prove this claim, suppose that $w\in \hubs(W^b) \cap \bigcup_{i=1}^t Z_i$, and let $T$ be its safe tree.
Consider an edge $e$ on the path from $w$ to $v$ in $T$. 
By the definition of a safe tree, $\cutof{X} = \cutof{\treecutout{T}{w}{e}}$ is a $K$-cut.
But $X\subsetneq V(T) = W^b$ and, since $v\in X$, we also have $X\not\subseteq \bigcup_{i=1}^t Z_i$,
contradicting Observation~\ref{obs:basic-laminar rooted}.


\begin{claim}\label{cla: hub-type1, w = v}
$v\in \hubs(W^b)$ iff $v \in \bigcap_{i=1}^t \hubs(Z_i)$.
\end{claim}

We start with the $(\Leftarrow)$ implication. Assume that $v \in \bigcap_{i=1}^t \hubs(Z_i)$.
For all $i = 1,...,t$, consider the witness tree for $v$ and $Z_i$.  By Observation~\ref{obs:hub-edge} their union is a
safe tree rooted at $v$ and thus it is a witness tree for $v$ and $W^b$. 
Together with $v\in V(\cutof{W^b})$ this implies that $v$ is an in-hub for $W^b$. 

To prove the $(\Rightarrow)$ implication, suppose that $v\in \hubs(W^b)$ and let $T$ be its safe tree.
We will show that $v$ is an out-hub for each $Z_i$. Indeed, for
any edge $e=(v,u_e)$ from $v$ in $T$, $u_e\in Z_j$ for some $j$ and $V(\treecutout{T}{v}{e})\subseteq Z_j$, as
$\cutof{\treecutout{T}{v}{e}}$ is a $K$-cut not containing the whole $W^b$.

Now consider an arbitrary but fixed $i$, and the tree $T'$ consisting of all edges $e$ from $v$ to $Z_i$ in $T$
plus a union of all corresponding trees $\treecutout{T}{v}{e}$. This
$T'$ spans $Z_i\cup\{v\}$. For each $e=(v,u_e), u_e\in Z_i$, by
Observations~\ref{obs:hub-subtree} and Observation~\ref{obs:hub-edge},
both $(\treecutout{T}{v}{e},u_e)$ and
$(\treecutout{T}{v}{e}\cup\braced{e},v)$ are safe trees. Now $T'$ is
the union of all these trees and thus Observation~\ref{obs:hub-union}
implies that $T'$ is a witness tree for $v$ and $\cutof{Z_i}$ and that $w=v$ is an out-hub for $Z_i$. 


\begin{claim}\label{cla: hub-type1, w not in Wb}
Let $w\in V\setminus W^b$. Then 
$w\in \hubs(W^b)$ iff $v \in \bigcap_{i=1}^t \hubs(Z_i)$ and $w\in N(v)\setminus W^b$.
\end{claim}

The $(\Leftarrow)$ implication is trivial: if $v \in \bigcap_{i=1}^t \hubs(Z_i)$ then,
by Claim~\ref{cla: hub-type1, w = v}, $v$ is an in-hub for $W^b$.
Then the assumption that $w\in N(v)\setminus W^b$ implies that $w$ is an out-hub for $W^b$,
directly by definition.

It remains to prove the $(\Rightarrow)$ implication. Suppose that 
$w\in \hubs(W^b)$ and let $T$ be its witness tree. By the definition of $\hubs(W^b)$
we have $w\in V(\cutof{W^b})$. 

We argue first that the degree of $w$ in $T$ is $1$.
Otherwise, if the degree of $w$ in $T$ is at least $2$, let $e$ be an edge from $w$
such that $v\in\treecutout{T}{w}{e}$ (such an edge necessarily exists.)
Then $\treecutout{T}{w}{e}$ does not span whole $W^b$ as there
is another edge from $w$ to $W^b$ in $T$. Thus $v \in V(\treecutout{T}{w}{e}) \subsetneq W^b$, 
and $\cutof{\treecutout{T}{w}{e}}$ is a $K$-cut, a contradiction again with Observation~\ref{obs:basic-laminar rooted}.

Therefore the degree of $w$ in $T$ is $1$. Let this single edge from $w$ be $(w,x)$. Then
Observation~\ref{obs:hub-edge} implies that $(T,w)$ is a safe tree if and
only if $(\treecutout{T}{w}{e},x)$ is a safe tree. Then
$\treecutout{T}{w}{e}$ witnesses that $x$ is an in-hub for $W^b$.
But Claims~\ref{cla: hub-type1, w in Z's} and~\ref{cla: hub-type1, w = v} imply
that only $v$ can be an in-hub for  $W^b$, providing that $v \in \bigcap_{i=1}^t \hubs(Z_i)$,
completing the proof of the $(\Rightarrow)$ implication in Claim~\ref{cla: hub-type1, w not in Wb}.
\end{proof}

If $b$ is an internal node of Type~$0$, $W^b$ may have multiple in-hubs,
but, as we show in the lemma below, they still must be common hubs for all $Z_i$'s. 
The out-hubs of $W^b$ may either be neighbors of such in-hubs, or common out-hubs for all $Z_i$'s. 


\begin{lemma}
\label{lem:hub-type0}
Let $b$ be a Type-$0$ internal node of the cactus. Let $Z_1,\ldots,Z_t$
be the children of $W^b$. Then
\begin{equation}
\hubs(W^b) \;=\; \bar{H}\cup(N(\bar{H}\cap W^b)\setminus W^b),
	\quad\textrm{where}\; \textstyle \bar{H} = V(\cutof{W^b})\cap \bigcap_{i=1}^t \hubs(Z_i).  
	\label{eqn: hub formula type 0}	
\end{equation}
\end{lemma}

\begin{proof}
We first show the $(\supseteq)$ inclusion. Consider some $w\in\bar{H}$. Then $w$ is a hub for each set
$Z_i$. Consider the corresponding witness trees $T_i$ for $w$ and $Z_i$ for all $i=1,\ldots,t$. 
By Observation~\ref{obs:hub-edge} their union is a safe tree rooted at
$w$ and thus it is a witness tree for $w$ and $\cutof{W^a}$. Together
with $w\in V(\cutof{W^b})$ this implies that $w$ is an in-hub for $W^b$.  

Next, consider some $u\in N(\bar{H}\cap W^b)\setminus W^b$.
That is, $u \in V\setminus W^b$, and $u$ is a neighbor of some $w\in\bar{H}\cap W^b$. 
By the previous paragraph, $w$ is an in-hub for $W^b$, so
Observation~\ref{obs:hub-edge} implies that $u$ is an out-hub for $W^b$. 
This completes the proof of the $(\supseteq)$ inclusion.

To show the $(\subseteq)$ inclusion, fix any $w\in \hubs(W^b)$, and let $T$ be its witness tree.
Recall that $T$ spans $W^b\cup\braced{w}$, so any neighbor of $w$ in $T$ is in $W^b = \bigcup_{i=1}^t Z_i$.
We need to show that $w \in \bar{H} \cup N(\bar{H}\cap W^b)\setminus W^b$. To this end,
we consider the three cases below.

\smallskip
\noindent
\mycase{1} $w\in W^b$. We will show that in this case $w\in\bar{H}$.
Note that $w\in V(\cutof{W^b})$ by the definition of a hub, so we
need to show that $w \in \bigcap_{i=1}^t \hubs(Z_i)$.

We start with the following simple observation.
For any edge $(w,u)$ in $T$, $\cutof{\treecutout{T}{w}{e}}$ is a $K$-cut
and $\treecutout{T}{w}{e} \subseteq W^b\setminus\braced{w} \subsetneq  W^b$, which implies that
$\treecutout{T}{w}{e}$ is a subset of one of the sets $Z_i$. 

Now, fix some arbitrary index $i \in \braced{1,...,t}$. It remains to show that $w\in \hubs(Z_i)$.
Let $e_1,\ldots,e_q$ be all the edges in $T$ from $w$ to $Z_i$. 
We claim that the subtrees $\treecutout{T}{w}{e_j}$, $j=1,\ldots,q$ cover all the vertices of $Z_i$, 
possibly with the exception of $w$ in case when $w\in Z_i$. 
Indeed, consider any vertex $x\in Z_i\setminus\braced{w}$ and the
first edge $e'=(w,u')$ on the path from $w$ to $x$ in $T$. By the
previous observation, $u'\in Z_i$, as $\treecutout{T}{w}{e'}$ contains $x\in Z_i$. 
Using Observations~\ref{obs:hub-subtree}, \ref{obs:hub-edge}, and~\ref{obs:hub-union}, this implies that
$\bigcup_{j=1}^q(\treecutout{T}{w}{e_j}\cup\braced{e_j})$ is a witness
tree for $w$ and $\cutof{Z_i}$, so $w\in \hubs(Z_i)$, as needed.

\smallskip
\noindent
\mycase{2} $w\not\in W^b$ and the degree of $w$ in $T$ is at least $2$. 
We will show that in this case $w\in\bar{H}$.
The argument is essentially the same as in Case~1.
We have that  $w\in V(\cutof{W^b})$, by the definition of a hub, and
it remains to show that $w \in \bigcap_{i=1}^t \hubs(Z_i)$.

For any edge $(w,u)$ in $T$, $\cutof{\treecutout{T}{w}{e}}$ is a $K$-cut
and $\treecutout{T}{w}{e} \subseteq W^b\setminus\braced{u'} \subsetneq  W^b$, 
where $u'$ is any neighbor of $w$ in $T$ other than $u$.
This implies that $\treecutout{T}{w}{e}$ is a subset of one of the sets $Z_i$ ---
the same property that we had in Case~1.

Following the same argument as in Case~1, for any index $i$ we can obtain the witness
tree for $w$ and $\cutof{Z_i}$ by combining the branches of $T$
inside $Z_i$, showing that $w\in \hubs(Z_i)$. Since $i$ is arbitrary,
we conclude that $w \in \bigcap_{i=1}^t \hubs(Z_i)$.

\smallskip
\noindent
\mycase{3}
$w\not\in W^b$ and the degree of $w$ in $T$ is $1$. 
In this case, we show that $w\in N(\bar{H}\cap W^b)\setminus W^b$.
Let this single edge from $w$ in $T$ be $e = (w,x)$. Then
Observation~\ref{obs:hub-edge} implies that $(T,w)$ is a safe tree if and
only if $(\treecutout{T}{w}{e},x)$ is a safe tree. However, then
$\treecutout{T}{w}{e}$ witnesses that $x$ is an in-hub for $W^b$
and we have already shown that this in turn happens if and only if
$x\in\bar{H}$. Thus $w\in N(\bar{H}\cap W^b)\setminus W^b$, completing the proof. 
\end{proof}


\subsubsection{Constructing the Hub Sets for Non-Trivial Cycles}
\label{sec:constructing-hubs-cycles}

In the previous sub-section we characterized hubs associated with $K$-cuts $W^b$, for
nodes $b\in\cactusC_G$. In this sub-section, we assume that $K$ is even and
we analyze hubs associated with the other type of basic $K$-cuts, namely with
$K$-cuts $W^C$, for non-trivial cycles $C$ of $\cactusC_G$.
The structural properties we establish are analogous to those in
Lemma~\ref{lem: K-cuts, cycle}, where they were based on a congestion-$K$ spanning tree $\Tstar$.
Here, we need to show that similar properties can be derived based on safe trees instead.
The path of $\Tstar$ that traversed all basic $K$-cuts associated with $C$, identified in Lemma~\ref{lem: K-cuts, cycle}, 
will be represented here by two sub-paths emanating from the head node $a_C$ of $C$,
one clockwise and the other counter-clockwise on $C$. These sub-paths will be
referred to as the front and back spine of $C$.

\smallskip

We denote the nodes of $C$ by $a_0, \ldots, a_{\ell-1}, a_\ell = a_0$, in the order along $C$,
where $\ell\ge 3$. That is, the links of $C$ are exactly $(a_0,a_1),(a_1,a_2),...,(a_{\ell-1},a_0)$. 
We assume that $a_0 = a_\ell$ is the head node $a_C$ of $C$.
The children of $W^C$ are the sets $Z_i=W^{a_i}$, $i=1,\ldots,\ell-1$.  For convenience,
we also use notation $Z_0=Z_\ell=V(\cutof{W^C})\setminus W^C$.
We stress that this set $Z_0$ is not a full shore of $W^C$ (unlike the other sets $Z_i$), it only includes 
the endpoints of the edges of the cut $\cutof{W^C}$ that are outside $W^C$. 
As in Section~\ref{subsec: cactus for congestion K},
for each $i = 0,...,\ell-1$, the set $\cutof{Z_i}\cap\cutof{Z_{i+1}}$ is a half-$K$-cut in $G$ (contains exactly $K/2$ edges).

We now define the two spines for $C$, mentioned earlier:
\begin{itemize}[nosep]
\item
A \emph{back spine} is a path $w_s$,\ldots,$w_\ell$ in $G$, with $s\in\{2, \ldots,\ell\}$,
such that $w_s\in Z_s\cap\hubs(Z_{s-1})\cap\hubs(Z_s)$, $w_i\in Z_i\cap \hubs(Z_i)$ for each $i = s+1,\ldots,\ell-1$,
and $w_\ell\in Z_\ell$.
\item
Symmetrically, a \emph{front spine} is a path $w_0,...,w_s$ in $G$, with $s\in \{0,\ldots,\ell-2\}$,
such that $w_0\in Z_0$, $w_i\in Z_i\cap \hubs(Z_i)$ for each $i = 1,\ldots,s-1$, and $w_s\in Z_s\cap\hubs(Z_s)\cap\hubs(Z_{s+1})$.
\end{itemize}


\begin{observation}
\label{obs:spine-hubs}
If $w_s, \ldots, w_\ell$ is a back spine then $w_i\in\hubs(Z_{i-1})$ for all $i=s,\ldots,\ell$.
Symmetrically, if $w_0, \ldots,  w_s$ is a front spine then  $w_i\in\hubs(Z_{i+1})$ for all $i=0,\ldots,s$.
\end{observation}

\begin{proof}
Consider the case of a back spine. For $i=s$ the claim is included in
the definition of a back spine. For each $i \in\braced{s+1,...,\ell}$, by the definition of back
spines, $w_{i-1}$ is an in-hub for $Z_{i-1}$, so, since we also have $(w_{i-1},w_i)\in \cutof{Z_{i-1}}$,
Observation~\ref{obs:hub-edge} implies that $w_i$ is an out-hub for $Z_{i-1}$. The case of the front spine is symmetric.
\end{proof}


\begin{lemma}\label{lem:hub-cycle}
Let $C$ be a non-trivial cycle in the cactus $\cactusC_G$ with vertices $a_0, a_1, \ldots, a_{\ell-1},a_\ell = a_0$,
listed in their order around $C$.  Then
\begin{enumerate}[{\rm(i)},nosep]
\item
If $w_2, \ldots, w_\ell$ is a back spine then $w_{\ell-1}, w_\ell \in \hubs(W^C)$.
\item
If $w_0, \ldots, w_{\ell-2}$ is a front spine then $w_0, w_1 \in \hubs(W^C)$.
\item 
If $w_0, \ldots, w_{g-2}$ is a front spine and $w_{g+1}, \ldots, w_\ell = w_0$ is a back spine,
for some index $g\in\{2,\ldots,\ell-1\}$, then $w_0 \in \hubs(W^C)$.
\item
$\hubs(W^C)$ contains only the vertices included in rules~(i), (ii) and~(iii).
\end{enumerate}
\end{lemma}

\begin{figure}[ht]
	\begin{center}
		\includegraphics[width = 5.5in]{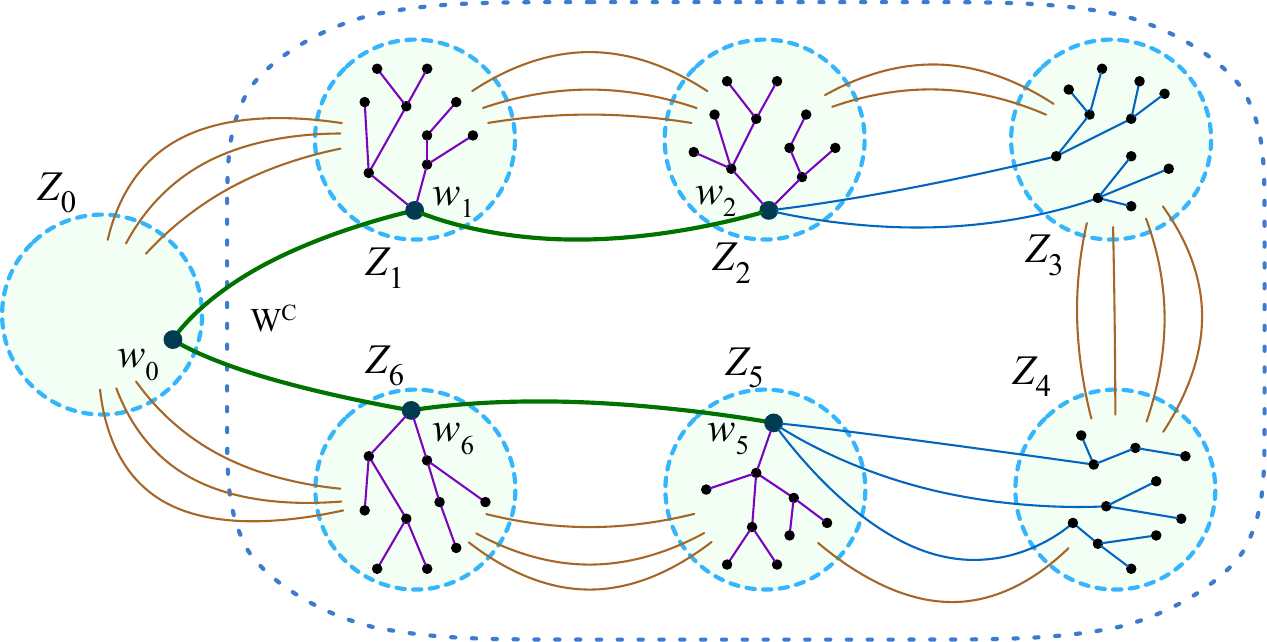}
	\end{center}
	\caption{An illustration of Lemma~\ref{lem:hub-cycle}(iii), for $K=8$, $\ell = 7$ and $g=4$.
	Thick blue edges show the front spine $w_0,w_1,w_2$ and back spine $w_5,w_6,w_7$, with $w_7 = w_0$.
	For $i\in\braced{1,2,5,6}$, the witness trees $T_i$ for $w_i$ and $Z_i$ are depicted with thin purple lines.
	The witness trees for $w_2$ and $Z_3$, and for $w_5$ and $Z_4$ are depicted with thin blue lines.
	Thin brown lines are non-tree edges in the cuts $\cutof{Z_i}$.}
	\label{fig: spines to witness tree}
\end{figure}

\begin{proof}
(i) Let $w_2$, \ldots, $w_\ell$ be a back spine. Let $T_1$ be
  a witness tree for $w_2$ and $Z_1$, and, for each $i=2,\ldots,\ell-1$,
  $T_i$ be a witness tree for $w_i$ and $Z_i$. These witness trees exist by the
  definition of a back spine. 
  Let $T'$ be the tree obtained as a union of trees $T_1,...,T_{\ell-1}$
  and edges $(w_1,w_2),...,(w_{\ell-2},w_{\ell-1})$,
  and $T''$ be the tree obtained by adding edge $(w_{\ell-1},w_\ell)$ to $T'$.
  From Observations~\ref{obs:hub-union}  and~\ref{obs:hub-edge}
  we obtain that $T'$ is a witness trees for $w_{\ell-1}$ and $W^C$,
  and $T''$ is a witness tree for $w_\ell$ and $W^C$. So $w_{\ell-1}, w_\ell \in \hubs(W^C)$.

The proof of (ii) is symmetric to (i).

(iii) 
Let $w_0, \ldots, w_{g-2}$ and $w_{g+1}, \ldots, w_{\ell-1}, w_\ell$ be a front and back spine,
where $g\in\{2,\ldots,\ell-1\}$. Then, from the definition of these spines, the following
witness trees exist:
\begin{itemize}[nosep]
\item a tree $T_i$ for $w_i$ and $Z_i$, for each $i\in\{1,\ldots,g-2\}$, and a tree $T_{g-1}$ for $w_{g-2}$ and $Z_{g-1}$, and
\item a tree $T_i$ for $w_i$ and $Z_i$, for each $i\in \{g+1,\ldots,\ell-1\}$  and a tree $T_{g}$ for $w_{g+1}$ and $Z_{g}$.
\end{itemize}
Let $T$ be a tree obtained as a union of all these trees, along with edges
$(w_0,w_1), ..., (w_{g-3},w_{g-2})$ and $(w_{g+1},w_{g+2}),...,(w_{\ell-1},w_{\ell})$.
(See an example in Figure~\ref{fig: spines to witness tree}.)
Using Observations~\ref{obs:hub-union} and~\ref{obs:hub-edge},
$T$ is a witness tree for $w_0=w_\ell$ and $W^C$, and thus $w_0 \in \hubs(W^C)$.

(iv) Assume that $w\in \hubs(W^C)$, and let $T$ be a witness tree for $w$ and $W^C$.
By the definition of hubs we have $w\in V(\cutof{W^C})$, so the structure of $C$
implies that $w\in Z_0\cup Z_1\cup Z_{\ell-1}$.

We first consider the case when $w\in Z_0=Z_\ell$, that is $w$ is an out-hub for $W^C$. The
argument relies on the two claims below.


\begin{claim}\label{cla: lem:hub-cycle, claim 1}
Let $i\in\{1,\ldots,\ell-1\}$. If $T\cap \cutof{Z_{i-1}}\cap\cutof{Z_{i}}\neq\emptyset$
and $T\cap \cutof{Z_{i}}\cap\cutof{Z_{i+1}}\neq\emptyset$ then 
all edges in $T\cap\cutof{Z_i}$ have a common endpoint, say $w_i$, and $w_i \in Z_i \cap \hubs(Z_i)$. 
\end{claim}

That, under the assumptions of the claim, all edges in $T\cap\cutof{Z_i}$ have
a common endpoint, follows from Lemma~\ref{lem:safe-tree-common-endpoint}.
Since $Z_{i-1} \cap Z_{i+1} = \emptyset$, we must have $w_i \in Z_i$.

It remains to show that $w_i \in \hubs(Z_i)$. Let $e_1 = (w_i,u_1),
..., e_k = (w_i,u_k)$ be all edges in $T$ such that $u_j\in Z_i$ for all $j = 1,...,k$.
Since $w\not\in Z_i$, the first part of the claim implies that 
for each $x\in Z_i\setminus\braced{w_i}$ the path from $w$ to $x$ goes through $w_i$.
So the trees $\treecutout{T}{w_i}{e_j}$, for $j = 1,...,k$, along with the singleton $\braced{w_i}$,
form a partition of $Z_i$.
For each $j$, by Observations~\ref{obs:hub-subtree} and Observation~\ref{obs:hub-edge},
both $(\treecutout{T}{v}{e_j},u_j)$ and $(\treecutout{T}{v}{e_j}\cup\braced{e_j},v)$ are safe trees.
Observation~\ref{obs:hub-union} now implies that
$\bigcup_{j=1}^k(\treecutout{T}{v}{e_j}\cup\braced{e_j})$ is a witness
tree for $w_i$ and $Z_i$, thus showing that $w_i \in \hubs(Z_i)$.


\begin{claim}\label{cla: lem:hub-cycle, claim 2}
Let $i\in\{1,\ldots,\ell-1\}$. Suppose that there is $v\in Z_{i-1}\cup Z_{i+1}$
that is a common endpoint of all all edges in $T\cap\cutof{Z_i}$. Then $v \in \hubs(Z_i)$. 
\end{claim}

The argument is similar to the one above. 
Let $e_1 = (v,u_1), ..., e_k = (v,u_k)$ be all edges in $T$ such that $u_j\in Z_i$ for all $j = 1,...,k$.
The assumption of the claim implies that
for each $x\in Z_i$ the path from $w$ to $x$ goes through $v$.
So the trees $\treecutout{T}{w_i}{e_j}$, for $j = 1,...,k$, form a partition of $Z_i$.
By the same reasoning as the one for  Claim~\ref{cla: lem:hub-cycle, claim 1},
we can conclude that $v \in \hubs(Z_i)$. 


In the next claim, we show that $T$ traverses all half-cuts of $C$, except one.

\begin{claim}\label{cla: lem:hub-cycle, claim 3}
There is exactly one index $g\in \{1,\ldots,\ell\}$ for which $T\cap \cutof{Z_{i-1}} \cap \cutof{Z_{i}}  = \emptyset$.
\end{claim}

That there is \emph{at least} one such index $g$, follows from Claim~\ref{cla: lem:hub-cycle, claim 1}:
If $T$ had an edge in each half-$K$-cut of $C$ then the edges in all half-$K$-cuts would form a cycle.
Since $T$ is a spanning tree of the subgraph induced by $W^C$, and since the half-$K$-cuts represented by $C$ 
contain all edges connecting different sets $Z_i$, the uniqueness of $g$ follows.

\smallskip

It remains to prove that the edges of $T$ in the half-$K$-cuts represented by $C$ form 
either a back spine satisfying condition~(i), or a front spine satisfying condition~(ii),
or can be divided into two spines that satisfy condition (iii).
With Claims~\ref{cla: lem:hub-cycle, claim 1}, \ref{cla: lem:hub-cycle, claim 2} and \ref{cla: lem:hub-cycle, claim 3},
this is just a matter of verifying that these conditions hold.

By Claim~\ref{cla: lem:hub-cycle, claim 1}, for each $i \in\braced{1,\ldots,\ell}\setminus\braced{g-1,g}$, 
all edges in $T\cap\cutof{Z_i}$ have a common endpoint $w_i\in Z_i$, and $w_i \in \hubs(Z_i)$ for $i\ne \ell$. Furthermore, if $g>1$ then
$w_{g-2}$ is a common endpoint of $T\cap\cutof{Z_{g-1}}\subseteq T\cap\cutof{Z_{g-2}}\cap\cutof{Z_{g-1}}$,
and Claim~\ref{cla: lem:hub-cycle, claim 2} implies that $w_{g-2} \in\hubs(Z_{g-1})$. 
Similarly, if $g<\ell$ then $w_{g+1}\in \hubs(Z_g)$.
Therefore:
\begin{itemize}[nosep]
\item If $g=1$, then $w_2,\ldots,w_{\ell-1},w_\ell$ is a back spine and $w=w_\ell$ satisfies~(i).
\item If $g=\ell$, then $w_0,w_1,\ldots,w_{\ell-2}$ is a front spine and $w=w_0$ satisfies~(ii).
\item If $g\in\{2,\ldots,\ell-1\}$ then $w=w_0,w_1,\ldots,w_{g-2}$ is
a front spine, $w_{g+1},\ldots,w_{\ell-1},w_\ell=w$ is a back
spine, and $w$ satisfies~(iii).
\end{itemize}
This completes the proof of~(iv) for the case when $w\in Z_0$.

The other case is when $w\in Z_1\cup Z_{\ell-1}$, that is $w$ is an in-hub for $W^C$. Recall that $T$ denotes
the witness tree for $w$ and $W^C$. By the definition of hubs and $Z_0$, $w$ has a neighbor $w_0\in Z_0$.
By Observation~\ref{obs:hub-edge}, $w_0\in \hubs(W^C)$ (that is, $w_0$ is an out-hub for $W^C$) and
$(T\cup\braced{(w_0,w)},w_0)$ is its witness tree for $w_0$ and $W^C$. The proof above for the case of
out-hubs now implies that $w_0$ satisfies either condition~(i) or~(ii); in the first case
$w$ satisfies~(i) as $w_{\ell-1}$, and in the second case $w$ satisfies~(ii) as $w_1$.
\end{proof}


\subsection{The Algorithm}
\label{subsec: the algorithm}

As explained at the beginning of this section, to determine whether $\stc{G} = K$ it is sufficient to
compute the hub sets for the basic $K$-cuts of $G$. This is because $\stc{G} = K$ if and only if
$r$ (the root vertex of $G$, that has degree $K$) is a hub for its basic $K$-cut $\cutof{Z}$, where $Z = V\setminus\braced{r}$.

The algorithm follows a dynamic programming paradigm, processing all basic $K$-cuts
bottom-up along their tree structure, as defined earlier in this section.
As presented (see the pseudo-code in Algorithm~\ref{alg:main}), it only solves the decision version,
determining whether $\stc{G} = K$ or not.
If $\stc{G} = K$, a spanning tree of $G$ with congestion $K$ can be reconstructed 
by standard backtracking. 

For each basic $K$-cut $\cutof{Z}$, the algorithm constructs a set $H(Z)$ intended to contain exactly the hubs for $Z$.
These sets are computed using a recurrence relation established in Section~\ref{subsec: rooting G and cactus}.
We start with external nodes. If $Z = W^a$, for an external node $a = \phi(w)$ of $\cactusC_G$,
then, according to Lemma~\ref{lem:hub-leaf}, in $H(Z)$ we include $w$ and its neighbors (line~7).
If $Z = W^a$ for an internal node $a$, then $H(Z)$ is computed from the hub sets of its children,
using either the recurrence from Lemma~\ref{lem:hub-type1}, if $a$ is of Type~1 (lines~10-11),
or the recurrence from Lemma~\ref{lem:hub-type0}, if $a$ is of Type~0 (lines 14-15).
In both cases, this computation can be implemented efficiently using standard data structures.


\begin{algorithm}[t]
\caption{The main algorithm}
\label{alg:main}
\begin{algorithmic}[1]
  \State \textbf{Input:}
  Graph $G=(V,E)$ and its a cactus representation $\cactusC_G,\phi$
\If{there exists a node $b$ of $\cactusC_G$ such that
  $|\phi^{-1}(b)|>1$} 
 \textbf{output} NO 
 \EndIf

\State \textbf{choose} a root $r\in V$ of degree $K$ in $G$
  
\State \textbf{order} the basic $K$-cuts linearly so that each child precedes its parent

\For{each basic $K$-cut $\cutof{Z}$, in this ordering}
\Case{$Z=W^a$ for an external node $a=\phi(w)$ for $w \in V \setminus\braced{r}$}
\State $H(W^a)\gets \{w\} \cup N(w)$
\EndCase

\Case{$Z=W^a$ for a Type-1 internal node $a=\phi(w)$ for some $w\in V$} 
\Let $Z_1,\ldots,Z_t$ be the list of all the children of $W^a$
\If{$w\in H(Z_i)$ for all $i=1,\ldots,t$}
$H(W^a)\gets \braced{w}\cup(N(w)\setminus W^a)$
\Else\ 
$H(W^a)\gets\emptyset$
\EndIf
\EndCase

\Case{$Z=W^a$ for a Type-0 internal node $a$}
\Let $Z_1,\ldots,Z_t$ be the list of all the children of $W^a$
\State
$\widehat{H} \gets V(\cutof{W^a})\cap H(Z_1)\cap\cdots\cap H(Z_t)$
\State
$H(W^a)\gets\widehat{H}\cup(N(\widehat{H}\cap W^a)\setminus W^a)$
\EndCase

\Case{$Z=W^C$ for a non-trivial cycle $C$}
\State
apply Algorithm~\ref{alg:cycle}
\EndCase
\EndFor

\IfThenElse{$r\in H(V\setminus\braced{r})$}{\textbf{output} YES}{\textbf{output} NO}
\end{algorithmic}
\end{algorithm}


The last case, relevant only when $K$ is even, is when $Z= W^C$, for a non-trivial cycle $C$ of $\cactusC_G$
(the pseudo-code for this case is given separately in Algorithm~\ref{alg:cycle}).
In this case the algorithm applies the recurrence implicit in Lemma~\ref{lem:hub-cycle}.
In order to do this, for each vertex $w\in Z_1 \cup Z_0 \cup Z_{\ell-1}$
we need to identify back and/or front spines that, based on 
the cases (i), (ii) or (iii) from this lemma, would imply that $w$ should be added to $H(Z)$.

The challenge is to implement this process efficiently. One key observation here is that
(by definition and Observation~\ref{obs:spine-hubs}), every non-empty
suffix of a back spine is also a back spine, and the analogous property applies to front spines.
This means that it's sufficient to only compute the maximum spine lengths for each candidate vertex $w$, not the actual spines.
We achieve this using an embedded dynamic programming  procedure that processes $C$ in the two directions,
and computes $H(Z)$ in time $\tildeO(k|C|)$ (the number of edges represented by $C$).

To give more detail, let's consider the case of back spines (the computation for front spines is symmetric).
In this case it is sufficient to calculate, for each candidate vertex $w_i\in Z_i$ of a back spine, the
minimum value $s$ for which there exists a path $w_s, \ldots, w_i$
that is a potential prefix of a back spine, i.e., each $w_{i'}$ on this path is an
in-hub for $Z_{i'}$ and $w_s$ is also an out-hub for $Z_{s-1}$. This
value $s$ computed by the algorithm is denoted $S^-(w_i)$. 
It is sufficient to compute $S^-(w_i)$ for the vertices $w_i$ in the candidate
set $U^-_i\subseteq Z_i$ which, in accordance with the definition of a back spine
and Observation~\ref{obs:spine-hubs}, contains only in-hubs for $Z_i$
that are also out-hubs for $Z_{i-1}$ (line~5). Two border cases are treated differently:
for $i=\ell$, $U^-_\ell$ is restricted, in a natural way, to out-hubs for $Z_{\ell-1}$ in $Z_\ell$ and,
for $i=0$, we let $U^-_1=\emptyset$ for technical convenience (line~3).
The values $S^-(w_i)$ are
then computed by a dynamic program starting from $S^-(w_2)=2$ for all
candidates $w_2\in Z_2$, and then, for increasing $i$, setting
$S^-(w_i)$ to be the maximum of $S^-(w_{i-1})$ over the neighbors of
$w_i$, or to $i$ if there are no candidate neighbors (line~7).

Once the maximum spine lengths, back and front, are computed, the calculation
follows the rules from Lemma~\ref{lem:hub-cycle} in a straightforward way (lines 13-16).


\begin{algorithm}[ht]
\caption{The subroutine for cycles}
\label{alg:cycle}
\begin{algorithmic}[1]
  \State \textbf{Input:} Cycle $C$ in $\cactusC_G$ of length $\ell\geq
  3$
  \Let $Z_1,\ldots,Z_{\ell-1}$ be the children of $W^C$ ordered along the
  cycle $C$

  \LeftComment{computing the back spines}

  \State $U^-_1\gets\emptyset$; $U^-_\ell \gets H(Z_{\ell-1})\setminus W^C$
  \For{$i=2,3,\ldots,\ell$}
  \If{$i<\ell$}  $U^-_i \gets Z_i\cap H(Z_i)\cap H(Z_{i-1})$
  \EndIf
  
  \For{all $w\in U^-_i$}
  \State $S^-(w)\gets\max(\braced{i}\cup \braced{S^-(v)\mid v\in
      N(w)\cap U^-_{i-1}})$
  \Comment{for $i=2$, $S^-(w)=2$}
  \EndFor
  \EndFor
    
  \LeftComment{computing the front spines}

  \State $U^+_{\ell-1}\gets\emptyset$; $U^+_0 \gets H(Z_1)\setminus W^C$
  \For{$i=\ell-2,\ell-3,\ldots,1,0$}
  \If{$i\geq 1$}  $U^+_i \gets Z_i\cap H(Z_i)\cap H(Z_{i+1})$
  \EndIf
  
  \For{all $w\in U^+_i$}
  \State $S^+(w)\gets\max(\braced{i}\cup\braced{S^+(v)\mid v\in N(w)\cap U^+_{i+1}})$
  \EndFor
  \EndFor

  \LeftComment{computing the out-hubs in both front and back spines}
  \State  $H^0 \gets \braced{w\in U^+_0\cap U^-_\ell\mid S^+(w)+3\geq S^-(w)}$

  \LeftComment{computing the in-hubs}
  
  \State $H^-\gets\braced{w\in U^-_{\ell-1}\cap V(\cutof{W^C})\mid S^-(w)=2}$
  \State $H^+\gets\braced{w\in U^+_1\cap V(\cutof{W^C})\mid S^+(w)=\ell-2}$

  \LeftComment{computing the resulting set of hubs}
  \State
  $H(W^C)\gets H^0 
  \cup H^-\cup (N(H^-)\setminus W^C)
  \cup H^+\cup (N(H^+)\setminus W^C)$

\end{algorithmic}
\end{algorithm}


\paragraph{Correctness proof.}
To prove the correctness of the algorithm, we need to show that it correctly decides whether $\stc{G} = K$.
This follows directly from the claim below.

\begin{claim}
\label{cla:alg-correct}
Algorithm~\ref{alg:main} (with the subroutine in
Algorithm~\ref{alg:cycle}) computes the correct hub sets, that is, for
each basic $K$-cut $\cutof{Z}$ we have $H(Z)=\hubs(Z)$. 
\end{claim}

To prove Claim~\ref{cla:alg-correct}, we prove that $H(Z)=\hubs(Z)$ for all basic $K$-cuts $\cutof{Z}$
inductively, in the order in which the sets $H(Z)$ are calculated by the algorithm.

For basic $K$-cuts $Z = \cutof{W^b}$, where $b\neq\phi(r)$, the calculation of
$H(W^b)$ in Algorithm~\ref{alg:main} exactly follows the statements
for $\hubs(W^b)$ in Lemmas~\ref{lem:hub-leaf}, \ref{lem:hub-type1},
and~\ref{lem:hub-type0} for external nodes, Type-$1$ internal nodes, and Type-$0$
internal nodes, respectively. So the inductive claim follows.

For a nontrivial cycle $C$, as already explained earlier,
in the first part of Algorithm~\ref{alg:cycle} we calculate, for each $w = w_\ell\in H(Z_{\ell-1})\setminus W^C$,
the value $S^-(w)$ equal to the minimal $s$ such that a back
spine $w_s, \ldots, w_\ell$ exists.  Similarly, for each $w = w_0\in
H(Z_1)\setminus W^C$, the value $S^+(w)$ is the maximal $s'$ for which
a front spine $w_0, \ldots, w_{s'}$ exists.

It remains to explain the meaning and the computation of sets $H^0$, $H^-$ and $H^+$.
These are simply the hub sets corresponding to the three different cases in
Lemma~\ref{lem:hub-cycle}.

Consider the computation of $H^0$.  If for some $w=w_0=w_\ell$ we
have $s'+3\geq s$ then we select a back spine $w_s$, \ldots, $w_\ell$
and a front spine $w_0$, \ldots, $w_{s'}$ that are guaranteed to exist
by the previous paragraph.  We let $g=s-1$ and observe that $s'\geq
g-2$ and $g\in\{2,\ldots,\ell-1\}$. Now $w_0$, \ldots, $w_{g-2}$ is a
front spine, as it is a prefix of the front spine $w_0$, \ldots,
$w_{s'}$ above and $w_{g+1}$, \ldots, $w_\ell$ is a back
spine equal to the back spine above. Thus $w\in\hubs(W^C)$ by
Lemma~\ref{lem:hub-cycle}(iii). Algorithm~\ref{alg:cycle} calculates
$H^0$ as the set of precisely all these hubs $w$.

Consider the computation of $H^-$. (The case of $H^+$ is symmetric.)
Each $w=w_{\ell-1}\in H^-$ is in $V(\cutof{W^C})$, thus it has a
neighbor $w_\ell\in Z_\ell$. We also have $S^-(w_{\ell-1})=2$, which
now guarantees the existence of a back spine $w_2, \ldots, w_{\ell-1}$, $w_\ell$. 
It follows that $w_{\ell-1}\in\hubs(W^C)$ by
Lemma~\ref{lem:hub-cycle}(i). Thus Algorithm~\ref{alg:cycle}
calculates $H^-$ as the set of all in-hubs from
Lemma~\ref{lem:hub-cycle}(i), and $N(H^-)\setminus W^C$ is the set of
out-hubs from Lemma~\ref{lem:hub-cycle}(i).

This completes the proof of the inductive claim for cycles, namely that
$H(W^C)=\hubs(W^C)$ for each non-trivial cycle $C$. 
The proof of Claim~\ref{cla:alg-correct} is now complete.


\myparagraph{Running time.}
We now analyze the running time. Recall that $n =|V|$ is the number of vertices
of $G$ and $m = |E|$ is the number of edges. As the input graph is $K$-edge-connected,
each vertex has degree at least $K$, so $m = \Omega(Kn)$.
One key property behind our estimate of the running time is that
for each basic cut $\cutof{Z}$, its corresponding hub set $H(Z)$
satisfies $|H(Z)|\le 2K$. This follows directly from the fact that
$H(Z)\subseteq V(\cutof{Z})$, and $|\cutof{Z}| = K$. Furthermore, the set differences
that occur in Algorithms~\ref{alg:main}~and~\ref{alg:cycle} are also subsets of $V(\cutof{Z})$, so they can be computed in $\tildeO(K)$ time.

We first analyze Algorithm~\ref{alg:cycle}. As explained above, all the sets $U^-_i$,
$U^+_i$, $H^0$, $H^-$, $H^+$ have size $O(K)$. Thus the computation of
spines and of the sets $H^0$, $H^-$, $H^+$ involves $O(K\ell)$ operations
on integers in $\braced{1,...,n}$ and vertex identifiers. The subsequent computation
consists of a constant number of operations on sets of vertices of size
$K$. Since the total length of the cycles in $\cactusC_G$ is $O(n)$,
the overall time for all invocations of Algorithm~\ref{alg:cycle} is
$\tildeO(Kn)=\tildeO(m)$.

Turning to Algorithm~\ref{alg:main}, the initialization part runs in time
$\tildeO(m)$, including the construction of the cactus representation
$\cactusC_G$, $\phi$ (see~\cite{karger_panigrahi_near-time-time_cactus_2009}).
Within the dynamic programming process, for each node $a$ of $\cactusC_G$ of degree $d$,
processing a basic $K$-cut $\cutof{Z} = \cutof{W^a}$ 
involves $O(d)$ set operations, each on sets of vertices of size
$O(K)$. Since the number of links in $\cactusC_G$ is $O(n)$, the total
time for processing all these $K$-cuts is  $\tildeO(Kn)=\tildeO(m)$.

Combining the bounds for Algorithms~1 and~2, we conclude that
the total running time of our algorithm is $\tildeO(m)$. 
Together with Claim~\ref{cla:alg-correct} this completes the proof of
Theorem~\ref{thm: algorithm for k-edge-connected}.



\bibliographystyle{abbrv}
\bibliography{stc_references}

@article{lin_2025_stc_interval_graphs,
title = {The spanning tree congestion problem on interval graphs},
journal = {Discrete Applied Mathematics},
volume = {377},
pages = {147-153},
year = {2025},
issn = {0166-218X},
doi = {https://doi.org/10.1016/j.dam.2025.06.054},
author = {Lan Lin and Yixun Lin},
}

@techreport{Fleiner_Frank_cactus_mincuts_2009,
	title = {A quick proof for the cactus representation of mincuts},
	author = {T. Fleiner and A. Frank}, 
	year = {2009},
	number = {QP-2009-03},
	institution = {Egerv\'ary Research Group, Budapest}
}

@inproceedings{karger_panigrahi_near-time-time_cactus_2009,
  title={A near-linear time algorithm for constructing a cactus representation of minimum cuts},
  author={Karger, David R and Panigrahi, Debmalya},
  booktitle={Proc. Twentieth Annual ACM-SIAM Symposium on Discrete Algorithms {(SODA)}},
  pages={246--255},
  year={2009},
  organization={SIAM}
}

@inproceedings{sandeep_1986_optimal_tree_machines,
 author = {Bhatt, Sandeep and Chung, Fan and Leighton, Tom and Rosenberg, Arnold},
title = {Optimal Simulations of Tree Machines},
year = {1986},
isbn = {0818607408},
address = {USA},
booktitle = {Proceedings of the 27th Annual Symposium on Foundations of Computer Science {(FOCS)}},
pages = {274-282},
numpages = {9},
}

@article{simonson_1987_variation_min_cut_arrangement,
author = {Simonson, Shai},
year = {1987},
month = {12},
pages = {235-252},
title = {A Variation on the Min Cut Linear Arrangement Problem.},
volume = {20},
journal = {Mathematical Systems Theory},
}

@inproceedings{rosenberg_1988_graph_embeddings,
author={Rosenberg, Arnold L.},
editor={Reif, John H.},
title={Graph embeddings 1988: Recent breakthroughs, new directions},
booktitle={VLSI Algorithms and Architectures},
year={1988},
publisher={Springer New York},
address={New York, NY},
pages={160--169},
isbn={978-0-387-34770-7}
}

@article{khuller_1993_designing_multi_commodity_flow_tree,
title = {Designing multi-commodity flow trees},
journal = {Information Processing Letters},
volume = {50},
number = {1},
pages = {49-55},
year = {1994},
issn = {0020-0190},
author = {Samir Khuller and Balaji Raghavachari and Neal Young},
}

@article{seymour_1994_call_routing,
  title={Call routing and the ratcatcher},
  author={Seymour, Paul D. and Thomas, Robin},
  journal={Combinatorica},
  volume={14},
  number={2},
  pages={217--241},
  year={1994},
  publisher={Springer}
}

@article{cai_1995_tree_spanner,
author = {Cai, Leizhen and Corneil, Derek  G.},
title = {Tree spanners},
journal = {SIAM Journal on Discrete Mathematics},
volume = {8},
number = {3},
pages = {359-387},
year = {1995},
}

@article{fakete_2001_tree_spanner,
title = {Tree spanners in planar graphs},
journal = {Discrete Applied Mathematics},
volume = {108},
number = {1},
pages = {85-103},
year = {2001},
issn = {0166-218X},
author = {S\'andor P. Fekete and Jana Kremer},
}

@article{ostrovoskii_2004_minimal_congestion_tree,
title = {Minimal congestion trees},
journal = {Discrete Mathematics},
volume = {285},
number = {1},
pages = {219-226},
year = {2004},
author = {M.I Ostrovskii},
}

@article{hruska_2008_tree_congestion,
  title={On tree congestion of graphs},
  author={Hruska, Stephen W},
  journal={Discrete Mathematics},
  volume={308},
  number={10},
  pages={1801--1809},
  year={2008},
  publisher={Elsevier}
}

@article{kozawa_2009_stc_graphs,
title = {On spanning tree congestion of graphs},
journal = {Discrete Mathematics},
volume = {309},
number = {13},
pages = {4215-4224},
year = {2009},
author = {Kyohei Kozawa and Yota Otachi and Koichi Yamazaki},
}

@article{law_2009_congestion_duality,
  title={Spanning tree congestion: duality and isoperimetry; with an application to multipartite graphs},
  author={Law, Hiu-Fai and Ostrovskii, MI},
  journal={Graph Theory Notes NY},
  volume={58},
  pages={18--26},
  year={2010}
}

@inproceedings{otachi_2010_complexity_result_stc,
author = {Otachi, Yota and Bodlaender, Hans L. and Van Leeuwen, Erik Jan},
title = {Complexity results for the spanning tree congestion problem},
year = {2010},
booktitle = {Proc. 36th Int. Conference on Graph-Theoretic Concepts in Computer Science ({WG})},
pages = {3-14},
}

@PhdThesis{lowenstein_2010_in_the_complement_dominating_set,
  author = 	{L{\"o}wenstein, Christian},
  title = 	{In the Complement of a Dominating Set},
  year = 	{2010},
  school =  {Technische {Universit\"{a}t} at Ilmenau},
}

@article{kozawa_2011_stc_rook_graphs,
  title={Spanning tree congestion of rook's graphs},
  author={Kozawa, Kyohei and Otachi, Yota},
  journal={Discussiones Mathematicae Graph Theory},
  volume={31},
  number={4},
  pages={753--761},
  year={2011},
  publisher={De Gruyter Open}
}

@inproceedings{okamoto_2011_hardness_results_exp_algorithm_stc,
author = {Okamoto, Yoshio and Otachi, Yota and Uehara, Ryuhei and Uno, Takeaki},
title = {Hardness Results and an Exact Exponential Algorithm for the Spanning Tree Congestion Problem},
year = {2011},
booktitle = {Proc. 8th Annual Conference on Theory and Applications of Models of Computation ({TAMC})},
pages = {452-462},
numpages = {11}
}

@article{dragan_2011_spanner_in_sparse_graph,
title = {Spanners in sparse graphs},
journal = {Journal of Computer and System Sciences},
volume = {77},
number = {6},
pages = {1108-1119},
year = {2011},
issn = {0022-0000},
author = {Feodor F. Dragan and Fedor V. Fomin and Petr A. Golovach},
}

@article{bodlaender_2011_stc_k-outerplanargraphs,
title = {Spanning tree congestion of k-outerplanar graphs},
author = {Hans L. Bodlaender and Kyohei Kozawa and Takayoshi Matsushima and Yota Otachi},
journal = {Discrete Mathematics},
volume = {311},
number = {12},
pages = {1040-1045},
year = {2011},
}

@article{bodlaender_2012_parameterized_complexity_stc,
author = {Bodlaender, Hans and Fomin, Fedor and Golovach, Petr and Otachi, Yota and Leeuwen, Erik},
year = {2012},
month = {09},
pages = {1-27},
title = {Parameterized complexity of the spanning tree congestion problem},
volume = {64},
journal = {Algorithmica},
}

@article{kubo_2015_spanning_tree_small_parameter,
author = {Kubo K. and Yamauchi Y. and Kijima S. and Yamashita M.},
title = {Spanning tree congestion problem on graphs of small diameter},
journal = {RIMS Kokyuroku},
volume = {1941},
pages = {17-21},
year = {2015},
url = { http://www.kurims.kyoto-u.ac.jp/∼kyodo/kokyuroku/contents/pdf/1941-03.pdf. }
}

@Incollection{otachi_2020_survey_spanning_tree_congestion,
	author={Otachi, Yota},
	title={A Survey on Spanning Tree Congestion},
	bookTitle={Treewidth, Kernels, and Algorithms: Essays Dedicated to Hans L. Bodlaender on the Occasion of His 60th Birthday},
	year={2020},
	publisher={Springer International Publishing},
	address={Cham},
	pages={165--172},
	isbn={978-3-030-42071-0},
}

@article{luu_chrobak_2025_better_hardness_algo,
  author       = {Huong Luu and
                  Marek Chrobak},
  title        = {Better Hardness Results for the Minimum Spanning Tree Congestion Problem},
  journal      = {Algorithmica},
  volume       = {87},
  number       = {1},
  pages        = {148--165},
  year         = {2025},
  url          = {https://doi.org/10.1007/s00453-024-01278-5},
  doi          = {10.1007/S00453-024-01278-5},
  bibsource    = {dblp computer science bibliography, https://dblp.org}
}

@incollection{dinitz_etal_strukture_systemy_1976,
	author =  {E. A. Dinic and A. V. Karzanov and M. V. Lomonosov}, 
	title ={The structure of a system of minimal edge cuts of a graph},
	booktitle = {Issledovani po diskretnoi optimizacii (Studies in discrete optimization)},
	 editors = {A.A. Fridman},
	 publisher = {Nauka, Moscow},
	 year = {1976},
	 pages = {290--306}, 
	note = {In Russian}
}

@article{law_etal_congestion_of_planar_graphs_2013,
  title={Spanning tree congestion of planar graphs},
  author={Law, Hiu and Leung, Siu and Ostrovskii, Mikhail},
  journal={Involve, a Journal of Mathematics},
  volume={7},
  number={2},
  pages={205--226},
  year={2013},
  publisher={Mathematical Sciences Publishers}
}

@inproceedings{kolman_iwoca_2024,
  author = {Kolman, Petr},
  title = {Approximating Spanning Tree Congestion on Graphs with Polylog Degree},
  booktitle = {Proc. Int. Workshop on Combinatorial Algorithms ({IWOCA})},
  year = {2024},
  pages = {497--508},
  doi = {10.1007/978-3-031-63021-7_38},
}

@inproceedings{kolman_approximating_congestion_2025,
  author       = {Petr Kolman},
  title        = {Approximation of Spanning Tree Congestion Using Hereditary Bisection},
  booktitle    = {Proc. 42nd Int. Symposium on Theoretical Aspects of Computer Science,
                  {STACS}},
  pages        = {63:1--63:6},
  year         = {2025},
  url          = {https://doi.org/10.4230/LIPIcs.STACS.2025.63},
  doi          = {10.4230/LIPICS.STACS.2025.63}
}

@inproceedings{lampis_etal_parameterized_spanning_tree_congestion_2025,
  author       = {Michael Lampis and
                  Valia Mitsou and
                  Edouard Nemery and
                  Yota Otachi and
                  Manolis Vasilakis and
                  Daniel Vaz},
  title        = {Parameterized Spanning Tree Congestion},
  booktitle    = {50th Int. Symposium on Mathematical Foundations of Computer
                  Science ({MFCS}) },
  pages        = {65:1--65:20},
  year         = {2025},
  url          = {https://doi.org/10.4230/LIPIcs.MFCS.2025.65},
  doi          = {10.4230/LIPICS.MFCS.2025.65}
}

@InProceedings{chandran_et_al_spanning_tree_congestion_2018,
  author =	{Chandran, L. Sunil and Cheung, Yun Kuen and Issac, Davis},
  title =	{{Spanning tree congestion and computation of generalized {Gy\"{o}ri-Lov\'{a}sz} partition}},
  booktitle =	{45th International Colloquium on Automata, Languages, and Programming (ICALP)},
  pages =	{32:1--32:14},
  ISBN =	{978-3-95977-076-7},
  ISSN =	{1868-8969},
  year =	{2018},
  URL =		{https://drops.dagstuhl.de/entities/document/10.4230/LIPIcs.ICALP.2018.32},
  URN =		{urn:nbn:de:0030-drops-90361},
  doi =		{10.4230/LIPIcs.ICALP.2018.32}
}


\newpage
\appendix


\section{$\NP$-Hardness for Degree~$4$}
\label{sec: np-hardness degree 4}

In this section, as a warm-up for our $\NP$-completeness proof of problem $\problemSTC$ for graphs of degree $3$, we
show a simpler proof for graphs of degree at most $4$. That is, we prove the following theorem.

\begin{theorem}\label{thm: np-completeness for degree 4}
Problem $\problemSTC$ is $\NP$-complete for graphs of degree at most $4$.
\end{theorem}

The proof is by reduction from problem $\problemMPNSAT$,
the version of SAT defined in Section~\ref{sec: np-hardness for degree-3 graphs} (see also~\cite{luu_chrobak_2025_better_hardness_algo}).
We show how, given a boolean expression $\phi$ that is an instance of $\problemMPNSAT$, to compute in polynomial time a graph $G$
and a constant $K$ such that
\begin{description}
	\item{$(\ast)$} $\phi$ is satisfiable if and only if $\stc{G}\le K$.
\end{description}

Our construction builds on the ideas 
from~\cite{luu_chrobak_2025_better_hardness_algo,lampis_etal_parameterized_spanning_tree_congestion_2025};
in particular we use the concept of a graph with double-weighted edges, with weights $\dbweight{a}{b}$ that are
polynomial in the size of $\phi$. This is permitted by Lemma~\ref{lemma: double weight gadget}
and our construction of $\dwgadget(a,b)$ in Appendix~\ref{sec: double weight gadget}.
(Alternatively, for this proof we can use the double-weight gadget of degree $4$ constructed in~\cite{lampis_etal_parameterized_spanning_tree_congestion_2025}.)
For brevity, from now on we will refer to double weights simply as ``weights''. Similarly, the ``degree'' of a vertex refers
to its weighted degree. 

Let $\phi$ be the given instance of $\problemMPNSAT$ with $n$ variables and $m$ clauses, of which $m'$
clauses are 2N- or 2P-clauses. We take $K = 3m + 5$, and we convert $\phi$ into a graph $G$ as follows (see Figure~\ref{fig: construction of G}):
\begin{itemize}\setlength{\itemsep}{-0.03in}
\item For each variable $x$, create a vertex $x$ and for each clause $\kappa$ create a vertex $\kappa$.

\item For each vertex $v$ that is a variable, 2N-clause or 2P-clause, create three other corresponding vertices:
	\emph{root vertices} $r^1_v$, $r^2_v$ 	and a \emph{terminal vertex}  $t_v$, connected by
	three unweighted edges	 $(r^1_v,r^2_v)$, $(r^1_v,t_v)$, and  $(r^2_v,t_v)$.
\item Add $n+m'$ weight-$2$ edges arbitrarily so that the root vertices form a cycle called the \emph{root cycle}.
	The edges in the cycle will be called \emph{root-cycle edges}, and the edges connecting the	
	root cycle to terminal vertices are \emph{root-terminal edges}.	
\item For each variable $x$, add a \emph{root-variable} edge $(x,t_x)$ with weight $\dbweight{1}{K-5}$.
\item For each 2P-clause $\gamma$, add an edge $(\gamma,t_\gamma)$ with weight $\dbweight{1}{K-1}$, 
	and for each 2N-clause $\alpha$, add an edge $(\alpha,t_\alpha)$ with weight $\dbweight{2}{K-1}$. 
	Call these edges \emph{root-clause} edges.
\item For each clause $\kappa$, add an edge from $\kappa$ to each vertex representing a variable whose literal (positive or negative) appears in $\kappa$. 
	If $\kappa$ is a positive clause, these edges have weight $\dbweight{1}{K-2}$, and if $\kappa$ is a negative clause,
	these edges have weight $\dbweight{1}{K-3}$. Call such edges \emph{clause-variable} edges.
\end{itemize}

\begin{figure}[ht]
\begin{center}
	\includegraphics[width=4in]{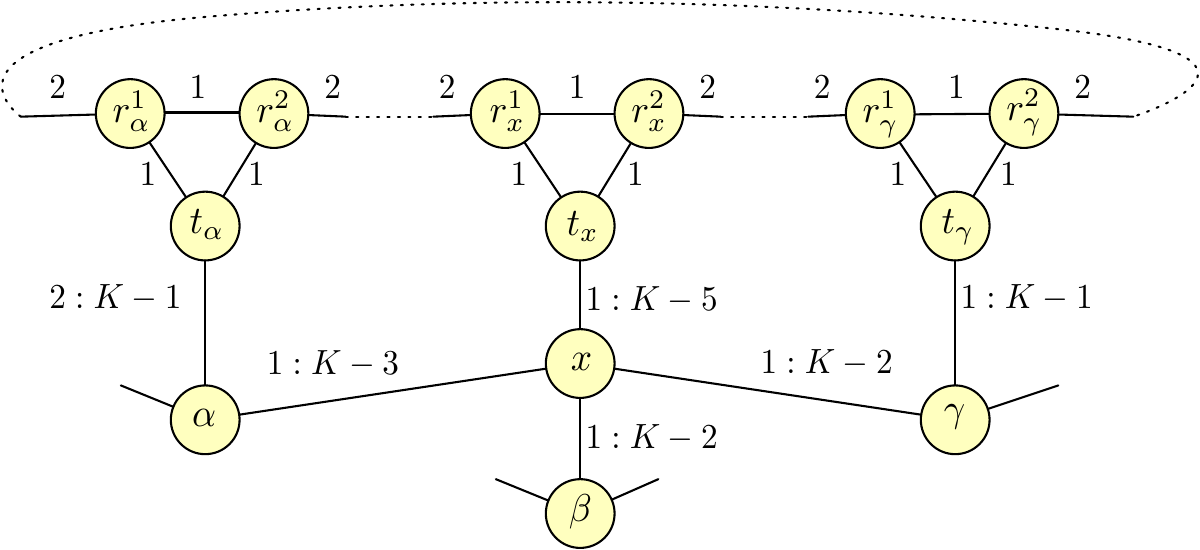}
    \caption{The construction of $G$.}
	\label{fig: construction of G}
\end{center}
\end{figure}

In the proof, when discussing $G$, for brevity, we will refer to the vertex representing a variable $x$ simply as
``variable $x$'' and, similarly, to the vertex representing a clause $\kappa$ as ``clause $\kappa$''.

It now remains to show that $G$ and $K$ satisfy condition~$(\ast)$. We prove the two implications in~$(\ast)$
separately. 


\medskip\noindent
{$(\Rightarrow)$} 
Given a truth assignment that satisfies $\phi$, we can construct a spanning tree $T$ for $G$ by adding to 
it the following edges (see Figure~\ref{fig: degree 4 assignment to tree}):
\begin{itemize}
\setlength{\itemsep}{-0.03in}
\item each root-variable edge,
\item for each clause $\kappa$, exactly one clause-variable edge from $\kappa$ to any variable whose literal satisfies $\kappa$
		(if $\kappa$ is satisfied by multiple literals, choose this variable arbitrarily),
\item all edges in the root cycle except for one edge of weight $2$,
\item for each variable or clause $u$, edge $(t_u,r^1_u)$. 
\end{itemize}

\begin{figure}[ht]
\begin{center}
	\includegraphics[width=4in]{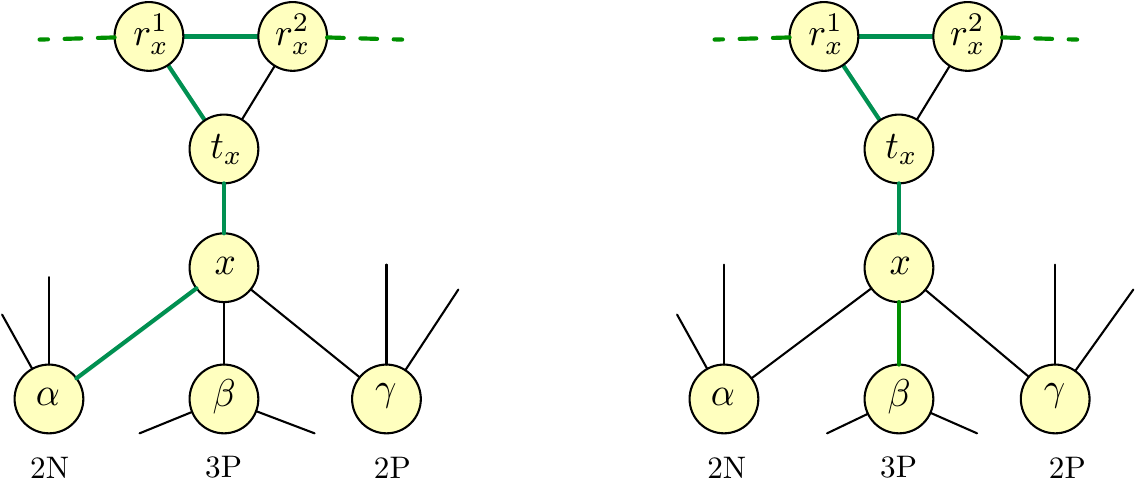}
    \caption{Converting a truth assignment satisfying $\phi$ into a spanning tree $T$. 
	The picture shows a variable $x$ with its 2N-clause $\alpha$, 3P-clause $\beta$ and 2P-clause $\gamma$.
	Tree edges are shown using thick (green) lines.
	On the left, the case when $x$ is false and edge $(x,\alpha)$ is chosen for $\alpha$.
	On the right, the case when $x$ is true and edge $(x,\beta)$ is chosen for $\beta$.
	Three other cases are not shown: when $x$ is chosen for $\gamma$,
	when $x$ is chosen for both $\beta$ and $\gamma$, and when $x$ is not chosen at all.
 	 }
	\label{fig: degree 4 assignment to tree}
\end{center}
\end{figure}

By inspection, $T$ forms a spanning tree, and each clause is a leaf in $T$. 
Furthermore, for each variable $x$, if $x$ has an edge in $T$ to its 2N-clause then $x$ does not have an edge in $T$ to any of its
positive clauses. 
To complete the proof of the $(\Rightarrow)$ implication,
we show that each edge $e$ in $T$ has congestion at most $K$. For this, we consider several cases. 

\smallskip\noindent\mycase{1}
$e$ is in the root cycle or it is a root-terminal edge. Consider first the sub-case when $e = (r^1_v, r^2_v)$, for some $v$ (variable or clause). 
Then $e$ contributes $1$ to its induced cut, $(r^2_v,t_v)$ contributes $1$, and the only edge from the root cycle not contained in $T$
contributes $2$.
All other edges in this cut are clause-variable and root-clause edges and there are at most $3m$ 
of them (because clauses are leaves in $T$). Thus, $\cng{T}{e} \le 4+3m  < K$. 

The second sub-case is when $e = (r^2_u,r^1_v)$, for some $u \neq v$ (variables or clauses), is very 
similar. Edge $e$ contributes $2$ to its cut, the only edge from the root cycle not contained in $T$
contributes $2$, and the clause-variable and root-clause edges contribute $3m$ at most,
so $\cng{T}{e} \le 4+3m < K$ as well.

In the third sub-case when $e$ is a root-terminal edge $(t_v,r^1_v)$, for some $v$,
$e$ contributes $1$ to its cut, and $(t_v,r^2_v)$ contributes $1$. If $v$ is a clause, then
the only other edge in the cut is a root-clause edge of weight at most $2$;
if $v$ is a variable, then all other edges in the cut are at most $5$ of the $9$ 
clause-variable edges from the clauses of $x$. Thus, $\cng{T}{e} \le 2+5  = 7 < K$.

\smallskip\noindent\mycase{2}
$e$ is a root-variable edge $e = (x,t_x)$ for some variable $x$. 
Let $\alpha, \beta, \gamma$ be the 2N-clause, 3P-clause, and 2P-clause of $x$, respectively. 
The shore $\treecutin{T}{x}{e}$ of $x$ of the cut induced by $e$ contains $x$ and can also contain some of its clauses, but if it contains
$\alpha$ then it contains none of  $\beta, \gamma$.
So $\treecutin{T}{x}{e}$ is either 
(i) $\braced{x}$, or 
(ii) $\braced{x,\kappa}$ for $\kappa \in \braced{\alpha, \beta, \gamma}$, or 
(iii) $\braced{x, \beta, \gamma}$. 
In case (i), $\cng{G,T}{e} = (K-5)+3 \le K-2$. 
In case (ii), $\cng{G, T}{e} \le (K-5) + 2 + 3  = K$. (It is equal $K$ only when $\kappa = \alpha$.) 
In case (iii), $\cng{G,T}{e} = (K-5) + 1 + 2 + 2 = K$. 

\smallskip\noindent\mycase{3}
$e$ is a clause-variable edge $e = (x, \kappa)$, for some clause $\kappa$ of $x$. 
If $\kappa = \alpha$ then $\cng{G,T}{e} = (K-3) + 3 = K$, and
if $\kappa \in \braced{\beta,\gamma}$ then $\cng{G,T}{e} = (K-2) + 2  = K$.


\medskip\noindent
{$(\Leftarrow)$} Let $T$ be a spanning tree with $\cng{G}{T}\le K$. We show how to convert $T$ into a satisfying assignment for $\phi$.
It is sufficient to prove the following lemma:

\begin{lemma}\label{lem: spanning tree to assignment}
Tree $T$ has the following properties:
\begin{description}\setlength{\itemsep}{-0.03in}
\item{(a)} $T$ does not contain any root-clause edges. 
\item{(b)} If a variable $x$ is connected in $T$ to its 2N-clause, then $x$ is not connected in $T$ to its 2P-clause or its 3P-clause.
\end{description}
\end{lemma}

The reason this lemma is sufficient is because it allows us to produce a satisfying assignment for $\phi$:
For each variable $x$, if $x$ is connected to its 2N-clause in $T$, make $x$ false; otherwise make it true. By 
Lemma~\ref{lem: spanning tree to assignment}(b), this is a valid truth assignment, and
by Lemma~\ref{lem: spanning tree to assignment}(a), every clause vertex $\kappa$ has an edge in $T$ to some of its variables, so
the definition of the truth assignment guarantees that this variable will satisfy $\kappa$.

\begin{proof}
The proof of  Lemma~\ref{lem: spanning tree to assignment}(a) is straightforward:
Suppose that a 2N- or 2P-clause $\kappa$ has its root-clause edge $f = (\kappa,t_\kappa)$ in $T$.
Let $x,y$ be the two variables of the literals in $\kappa$. Then
there are two disjoint paths from $\kappa$ to $t_\kappa$ in $G$ that do not use $f$:
namely paths starting with $\kappa - x - t_x$ and $\kappa-y-t_y$, and then following the root cycle to $t_\kappa$, for one path clockwise and
for the other one counter-clockwise, to ensure that the paths are disjoint.
The cut $\cutof{\sptreecut{T}{f}}$ induced by $f$ must contain at least one edge from each of these two paths,
implying $\cng{T}{f} \ge K-1 +2 = K+1$, a contradiction. 

\smallskip

The rest of this section is devoted to the proof of Lemma~\ref{lem: spanning tree to assignment}(b), which
will complete the proof of $\NP$-completeness. We start with the following claim:


\begin{claim}\label{cla: root cycle in spanning tree}
For any two root vertices, the path in $T$ between these vertices
consists of only root-cycle edges and root-terminal edges.
\end{claim}

To justify Claim~\ref{cla: root cycle in spanning tree}, consider any two different root vertices $p$ and $q$ and suppose that the $p$-to-$q$ path $P$ in $T$
does not satisfy Claim~\ref{cla: root cycle in spanning tree}. Since, by Lemma~\ref{lem: spanning tree to assignment}(a), $P$ does not contain any
root-clause edges, $P$ must contain at least one clause-variable edge, say $f=(y,\kappa)$
(cf. Figure~\ref{fig: cut edge between two root vertices}).
But then in the cut $\cutof{\sptreecut{T}{f}}$ induced by $f$, the vertices $p$ and $q$ would be in different shores, 
so this cut would have to contain two root-cycle edges, or a root-cycle edge and two root-terminal edges, or
four root-terminal edges; in any case, the total weight of these edges is $4$. 
Thus, $\cng{T}{f} \ge K-3+ 4 = K+1$ -- contradicting the assumption that $\cng{G}{T}\le K$. 
So Claim~\ref{cla: root cycle in spanning tree} holds.


\begin{figure}[ht]
\begin{center}
	\includegraphics[width=4in]{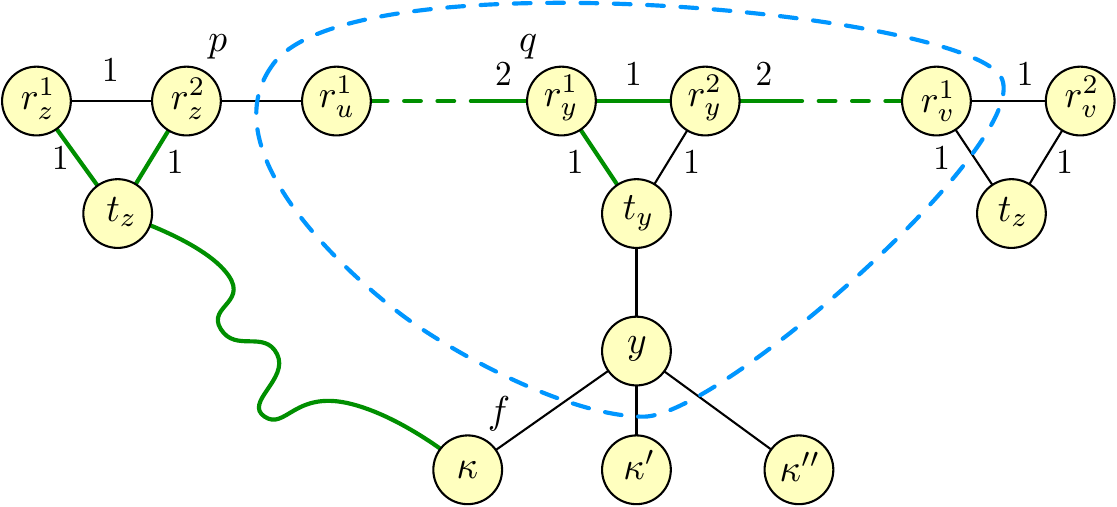}
    \caption{The cut of a clause-variable edge $f = (\kappa,y)$
	 that is on a path between two different root vertices $p$ and $q$, implying $\cng{T}{f} \ge K+1$.
	 Tree edges are shown using thick (green) lines. The cut is indicated using a (blue) dashed curve.}
	\label{fig: cut edge between two root vertices}
\end{center}
\end{figure}


\smallskip

Continuing with the proof of Lemma~\ref{lem: spanning tree to assignment}(b),
let $x$ be any variable, and let $\alpha,\beta,\gamma$ be its 2N-clause, 3P-clause, and 2P-clause. 
Assume for contradiction that edges $(\alpha,x)$ and $(\pi, x)$ are in $T$ for some $\pi\in\braced{\beta,\gamma}$. 
Consider the path $P$ in $T$ connecting $\braced{\alpha,x,\pi}$ to the root cycle.
By Lemma~\ref{lem: spanning tree to assignment}(a), $P$ must use a root-variable
edge, and let $e = (y,t_y)$ be the first such edge. We will examine the
edges crossing the cut $\cutof{\sptreecut{T}{e}}$ induced by $e$.
From Claim~\ref{cla: root cycle in spanning tree},  we obtain:


\begin{corollary}\label{cor: root chain is on one side}
The root cycle is on the shore $\treecutin{T}{t_y}{e}$ of $\cutof{\sptreecut{T}{e}}$, 
while $\alpha,x,\pi$ are all on the shore $\treecutin{T}{y}{e}$.
\end{corollary}

By the definition of $\problemMPNSAT$, different clauses in $\phi$ cannot share more than one variable. 
This implies that there are four distinct variables $z_1, z_2, z_3, z_4$, none equal to $x$, 
with clause-variable edges from $\alpha,\beta,\gamma$, namely
such that $z_1 \in \alpha$, $z_2, z_3 \in \beta$, and $z_4 \in \gamma$. 
This in turn implies that there are $7$ edge-disjoint paths in $G$
from the set $\braced{\alpha,x,\pi}$ to the root cycle. These paths are:
\begin{itemize}\setlength{\itemsep}{-0.03in}
    \item $\alpha - t_{\alpha} - r^1_\alpha$, $\alpha - t_{\alpha}- r^2_\alpha$ (these count as disjoint because $w_1({\alpha, t_\alpha}) = 2$),
		$\alpha - z_1 - t_{z_1} - r^1_{z_1}$,
	\item $x-t_x-r^1_x$,
    \item if $\pi = \beta$: \ 
		 $\beta - z_2 - t_{z_2} - r^1_{z_2}$, $\beta - z_3 - t_{z_3} - r^1_{z_3}$, $x- \gamma - t_{\gamma} - r^1_{\gamma}$.
	\item if $\pi = \gamma$: \  
	 	$\gamma - t_{\gamma} - r^1_{\gamma}$,  $\gamma - z_4 - t_{z_4} - r^1_{z_4}$, $x-\beta - z_2 - t_{z_2} - r^1_{z_2}$.
\end{itemize}
All these $7$ paths cross cut $\cutof{\sptreecut{T}{e}}$, and at most one of them can use edge $e$ to cross it. 
So the congestion of $e$ is at least
  $\cng{T}{e} \ge  K-5 + 6 = K+1$, contradicting the choice of $T$,
  and completing the proof of Lemma~\ref{lem: spanning tree to assignment}.
\end{proof}


\section{The Double-Weight Gadget -- Proof of Lemma~\ref{lemma: double weight gadget}}
\label{sec: double weight gadget}


In this section we prove Lemma~\ref{lemma: double weight gadget}, by
describing the construction of our double-weight gadget $\dwgadget(a,b)$,
that will be used to replace edges with double weights in  our
$\NP$-completeness proof in Section~\ref{sec: np-hardness for degree-3 graphs}.
The double-weight gadget is constructed from a sub-gadget that we call a \emph{bottleneck},
that we introduce first.


\myparagraph{The bottleneck gadget.}
For an integer $w\ge 3$, we define a \emph{$w$-bottleneck of degree $3$} as an unweighted graph $\bottleneck = \bottleneck(w)$ 
with two distinguished vertices $s,t$ called the \emph{gates of $\bottleneck$}, that has the following properties:
\begin{description}[nosep]
\item{(b1)} $\stc{\bottleneck} = w$, 
\item{(b2)} for any spanning tree $T$ of $\bottleneck$, the $s$-to-$t$ path in $T$ contains an edge $e$ with $\cng{\bottleneck , T}{e}\ge w$, 
\item{(b3)} the degrees of $s$ and $t$ are $2$ and all other degrees in $\bottleneck$ are at most $3$.
\end{description}

We now show a construction of $\bottleneck(w)$. We remark that congestion properties of some grid-like graphs of degree $3$ 
have been analyzed in the literature, for example in \cite{law_etal_congestion_of_planar_graphs_2013} for hexagonal grids.
However, these results are not sufficient for our purpose because they don't explicitly address condition~(b2) of bottleneck graphs.

Let $w\ge 3$ be the specified parameter.
Our gadget is a hexagonal grid, depicted as a slanted wall-of-bricks grid of dimensions $w\times w$,
illustrated in Figure~\ref{fig: wall-of-bricks gadget} for $w = 4$ and $w = 5$.
More precisely, geometrically the grid has $w-1$ rows of $2\times 1$ rectangular faces called \emph{bricks}, with $w-1$ bricks per row.
The vertices are the corners of the rectangles, with the edges represented by lines connecting these corners. 
The gate $s$ of $\bottleneck$ is the bottom left corner while $t$ is chosen to be the top right corner. 
The gates have degree 2 and all other vertices have degree 2 or 3, so property~(b3) is satisfied.

\begin{figure}[ht]
\begin{center}
	\includegraphics[width=6in]{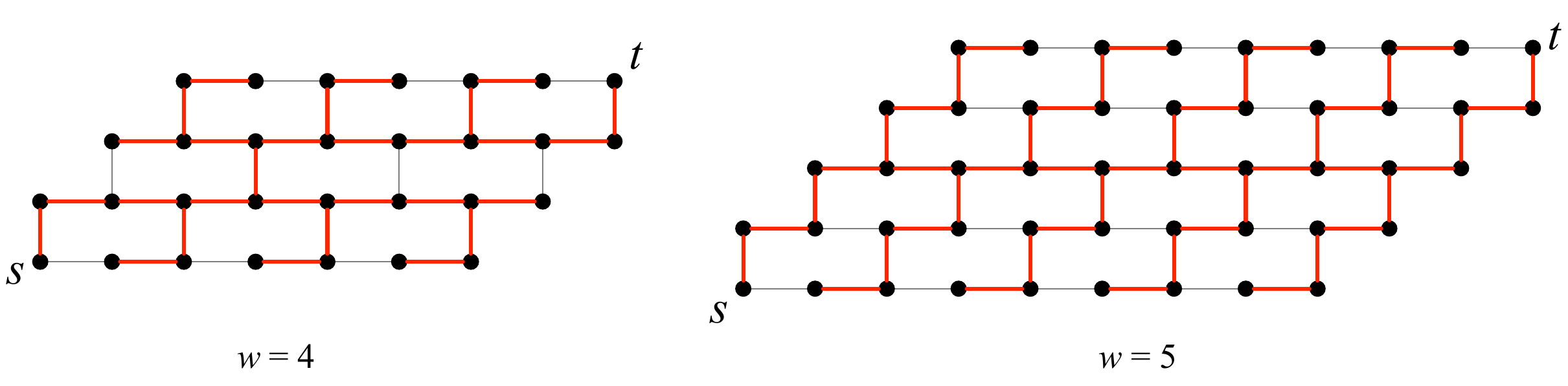}
   	 \caption{Bottleneck gadget $\bottleneck$ for $w=4$ and $w=5$. Their spanning trees $\Tstar$ with congestion $w$ are shown in thick (red).}
	\label{fig: wall-of-bricks gadget}
\end{center}
\end{figure}

First, we claim that $\bottleneck$ has property~(b1) that $\stc{\bottleneck} = w$.
It is sufficient to prove that $\stc{\bottleneck} \le w$, because the other inequality follows from~(b2).
To this end, we specify a spanning tree $\Tstar$ of $\bottleneck$ with congestion $w$ (See Figure~\ref{fig: wall-of-bricks gadget}.)  
For odd $w$, $\Tstar$ includes the middle horizontal path, with vertical branches
alternatively connecting to the top and bottom boundaries. 
For even $w$, $\Tstar$ includes both middle horizontal paths, with one of the two middle vertical edges connecting them. 
The vertical branches are analogous to the odd case.
By inspection, the congestion of $\Tstar$ in both cases is $w$. Specifically,
for odd $w$, excluding the edges on the left and right boundaries,
the edges in the middle horizontal path and the vertical edges touching this path have congestion $w$.
All other edges have lower congestion.
For even $w$, the connecting edge in the middle row has congestion $w$,
and the edges along the middle horizontal paths touching the central brick also have congestion $w$.
All other edges have lower congestion.
Note that in both cases, the $s$-to-$t$ path in $\Tstar$ includes some of congestion-$w$ edges. 

Next, we prove property~(b2), namely
that \emph{for any spanning tree} $T$ of $\bottleneck$, at least one of the edges in the $s$-to-$t$ path $P$ in $T$ has congestion at least $w$. 
To do this, it is convenient to express the argument in terms of the dual graph. 
This follows the approach in~\cite{hruska_2008_tree_congestion,law_2009_congestion_duality},
where the authors studied congestion in rectangular and triangular grid graphs. 
Let $\bottleneck_{st} = \bottleneck + (s,t)$, a graph obtained from $\bottleneck$ by adding an edge $(s,t)$,
that splits the outer face of $\bottleneck$ into two faces. $T$ remains a valid spanning tree of $\bottleneck_{st}$. 
Let $\dual{\bottleneck}_{st}$ be the dual graph of $\bottleneck_{st}$.  
($\dual{\bottleneck}_{st}$ has some parallel edges, because $\bottleneck$ has some vertices of degree $2$.)
Define $z_{i,j}$ to be the dual vertex in $\dual{\bottleneck}_{st}$ representing the brick at 
(skewed) column $i$ and row $j$, where $1 \le i,j \le w-1$. 
We also define two special dual vertices: $z_{w,0}$ is the dual vertex of the outer face whose boundary 
includes the bottom and right boundaries of $\bottleneck$, and $z_{0,w}$ corresponds to the outer face whose boundary 
include the top and left boundaries of $\bottleneck$ (See Figure~\ref{fig: wall-of-bricks dual}.)

\begin{figure}[ht]
\begin{center}
	\includegraphics[width=3.25in]{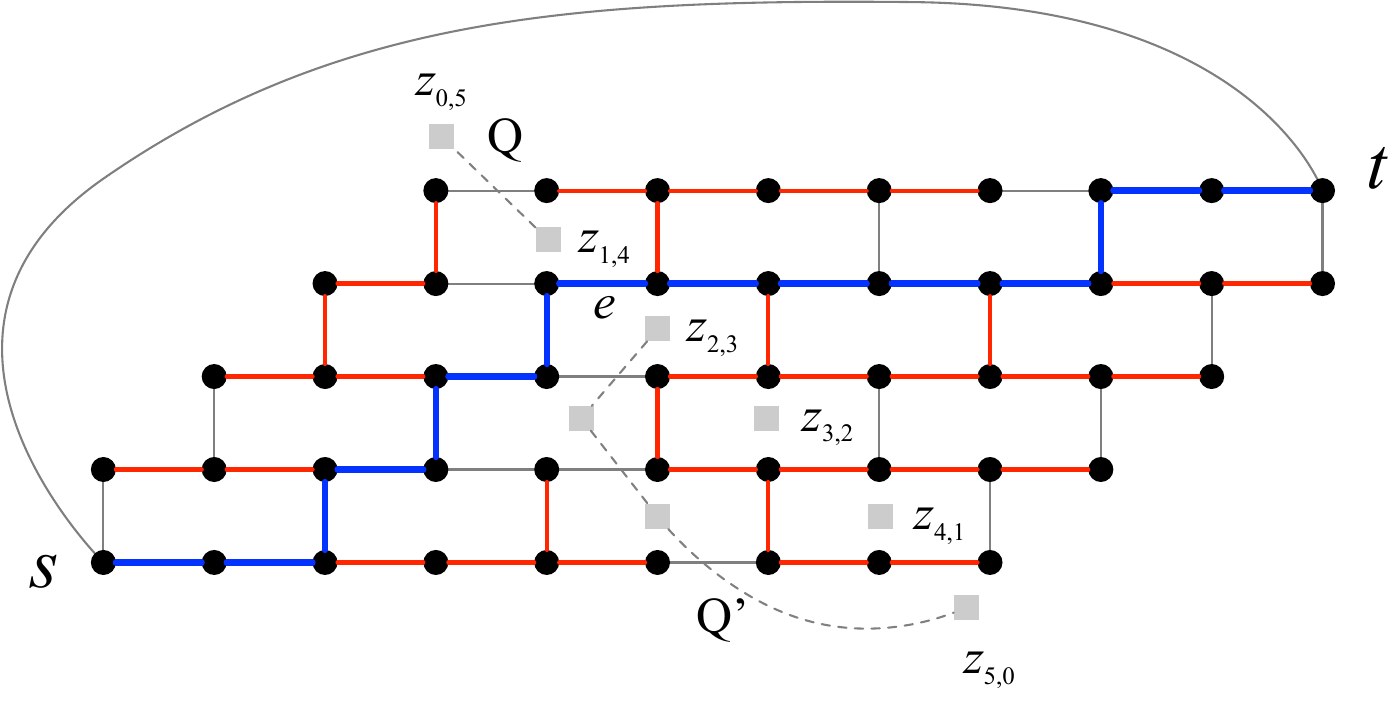}
   	 \caption{Illustration of an $s$-$t$ path $P$ in thick (blue) edges that includes an edge $e$ with congestion $w=5$.}
	\label{fig: wall-of-bricks dual}
\end{center}
\end{figure}
Let $\dual{T}$ be the graph obtained from $\dual{\bottleneck}_{st}$ by removing the edges dual to the edges of $T$.
Then $\dual{T}$ is connected and acyclic, so it is a spanning tree of $\dual{\bottleneck}_{st}$, and is referred to as the
tree dual to $T$.

Now consider the dual vertices corresponding to the bricks on the main diagonal of $\bottleneck$ namely the vertices  $z_{i,j}$ where $i+j=w$. 
Let $C$ be the set of (horizontal) edges of $\bottleneck$ dual to the edges of the path 
$z_{0,w} - z_{1,w-1} - z_{2, w-2} - \cdots - z_{w-1, 1} - z_{w,0}$ in $\dual{\bottleneck}_{st}$.
$C$ is a cut of $\bottleneck$ (consisting of the edges along the main diagonal), and therefore $P$ contains at least one edge from $C$, say 
an edge $e \in T$ whose dual is $\dual{e} = (z_{p,q}, z_{p+1,q-1})$. Note $z_{p,q}$ is above $e$ while $z_{p+1, q-1}$ is below it. 
Since the cycle $P + (s,t)$ encloses the top-left region of $\bottleneck_{st}$ (more precisely, the corresponding dual edges
form a cut of $\bottleneck$ that separates $z_{0,w}$ from $z_{w,0}$), there are two disjoint paths $Q, Q'$ in $\dual{T}$ that
start at $z_{p,q}$ and $z_{p+1,q-1}$ and end at $z_{0,w}$ and $z_{w,0}$, respectively (See Figure~\ref{fig: wall-of-bricks dual} for an example). 
Let $l(Q), l(Q')$ denote their respective lengths. 
Both $Q, Q'$ must cross the boundaries of $\bottleneck$. 
If the two boundaries they cross are top and bottom then $l(Q) + l(Q') \ge w - q + (q-1)  = w-1$;
if they are left and right then $l(Q) + l(Q') \ge p + w - (p+1) = w-1$. 
Otherwise, if the two boundaries are top and right, then $l(Q) + l(Q') \ge w - q + w - (p+1)  = w-1$; 
if they are left and bottom then $l(Q)+l(Q') \ge p + q - 1 = w-1$. 
In all cases, $l(Q)+l(Q') \ge w-1$. 
It is not difficult to see that the congestion of $e$ is equal to $l(Q)+l(Q')+1$; 
for a detailed proof of this observation see the analyses of rectangular grids in~\cite[Section~3.3]{hruska_2008_tree_congestion}
and triangular grids in~\cite[pp.~6]{law_2009_congestion_duality}, that both naturally apply also to our wall-of-bricks gadget.
This means  $\cng{G,T}{e} \ge  w$ which completes the proof for property (b2).


\myparagraph{Double-weight gadget.}
To construct \emph{a double-weight gadget $\dwgadget(a,b)$ of degree $3$}, where $\dbweight{a}{b}$ is a double weight with $a<b$, 
create $a$ disjoint copies $\bottleneck_1,...,\bottleneck_a$ of $\bottleneck(b-a+1)$ (constructed above), 
where the gates of each
$\bottleneck_i$ are $s_i$ and $t_i$. Then add two additional vertices $s^\ast$ and $t^\ast$ called the \emph{ports of $\dwgadget(a,b)$},
with $s^\ast$ connected by edges to each gate $s_i$ and $t^\ast$ is connected by edges to each gate $t_i$.
(See Figure~\ref{fig: double weight gadget}.)

Clearly $\dwgadget(a,b)$ satisfies Lemma~\ref{lemma: double weight gadget}(i). That this gadget satisfies Lemma~\ref{lemma: double weight gadget}(ii) was proved by Lampis~et~al.~\cite[Lemma~10]{lampis_etal_parameterized_spanning_tree_congestion_2025} when $\bottleneck$ is a square-grid,
but their argument trivially extends to any bottleneck gadget, so we omit the proof here.



\begin{figure}[ht]
	\begin{center}
		\includegraphics[width = 6in]{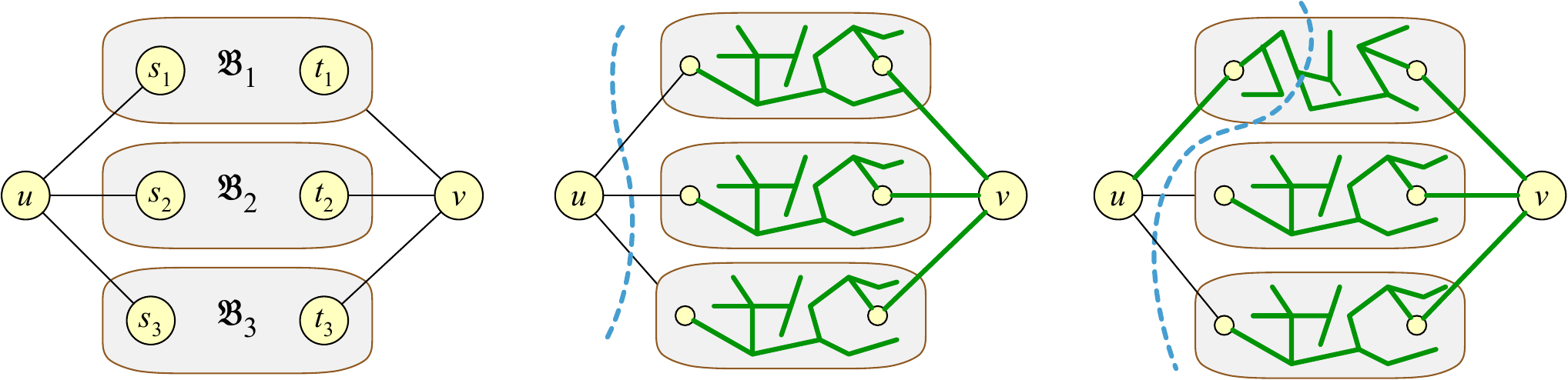}
	\end{center}
	\caption{On the left, the double-weight gadget $\dwgadget(a,b)$ replacing an edge $(u,v)$ 
	of $G$ with weight $\dbweight{a}{b}$ to produce a new graph $G'$. Here, $a=3$.
	The pictures in the middle and right illustrate how this gadget ``simulates'' the double weight.
	If $(u,v)$ is not in the spanning tree of $G$, the corresponding spanning tree of $G'$ can traverse the gadget as in the middle
	picture, contributing $a$ to the congestion.
	If $(u,v)$ is in the spanning tree of $G$, the corresponding spanning tree of $G'$ can traverse the gadget as
	in the picture on the right, contributing $(b-a+1) + a-1 = b$ to the congestion.
	}
	\label{fig: double weight gadget}
\end{figure}

\end{document}